\documentclass[11pt, a4paper]{article}

\usepackage[margin=1in]{geometry} 

\usepackage[T1]{fontenc}
\usepackage{lmodern}
\usepackage[english]{babel}

\usepackage{amsmath}
\usepackage{amssymb}
\usepackage{amsthm}
\usepackage{bm}
\usepackage{mathrsfs} 

\usepackage{array}
\usepackage{multirow}
\usepackage{enumitem}     
\usepackage{booktabs}

\usepackage{caption}
\usepackage{subcaption}
\usepackage{graphicx}
\usepackage{xcolor}

\usepackage{tikz}
\usetikzlibrary{positioning, shapes.geometric, arrows.meta}
\usepackage{pgfplots}
\pgfplotsset{compat=1.17}

\usepackage{algorithm}
\usepackage{algorithmic}
\let\oldalgorithmiccomment\algorithmiccomment
\renewcommand{\algorithmiccomment}[1]{\hfill \oldalgorithmiccomment{#1}}

\usepackage{hyperref}

\usepackage{comment}

\theoremstyle{definition}
\newtheorem{definition}{Definition}[section]

\theoremstyle{plain}

\newtheorem{assumption}{Assumption}

\newtheorem{lemma}{Lemma}

\newtheorem{theorem}{Theorem}

\title{Computing Maximal Per-Record Leakage and Leakage-Distortion Functions \\ for Privacy Mechanisms under Entropy-Constrained Adversaries}

\author{
	Genqiang Wu\thanks{School of Information Engineering and Artificial Intelligence, Lanzhou University of Finance and Economics; Key Laboratory of Smart Commerce of Gansu Province; Corresponding author. Email: wahaha2000@icloud.com.} 
	\and Xiaoying Zhang\thanks{School of Economics, Lanzhou University of Finance and Economics. Email: zxynk2013@126.com.}
	\and Yu Qi\thanks{Chongqing Branch, China Mobile Group Design Institute Co., Ltd. Email: qiyu@cmdi.chinamobile.com.}
	\and Hao Wang\thanks{School of Information Engineering and Artificial Intelligence, Lanzhou University of Finance and Economics. Email: 2488169785@qq.com.}
	\and Jikui Wang\thanks{School of Information Engineering and Artificial Intelligence, Lanzhou University of Finance and Economics. Email: wjkweb@163.com.}
	\and Yeping He\thanks{Institute of Software, Chinese Academy of Sciences. Email: yeping@nfs.iscas.ac.cn} 
	}

\DeclareMathOperator*{\argmax}{arg\,max}

\begin{document}
	
\maketitle


\begin{abstract}
	The exponential growth of data collection necessitates robust privacy protections that preserve data utility. We address information disclosure against adversaries with bounded prior knowledge, modeled by an entropy constraint $H(X) \geq b$. Within this information privacy framework---which replaces differential privacy's independence assumption with a bounded-knowledge model---we study three core problems: maximal per-record leakage, the primal leakage-distortion tradeoff (minimizing worst-case leakage under distortion $D$), and the dual distortion minimization (minimizing distortion under leakage constraint $L$).
	
	These problems resemble classical information-theoretic ones (channel capacity, rate-distortion) but are more complex due to high dimensionality and the entropy constraint. We develop efficient alternating optimization algorithms that exploit convexity-concavity duality, with theoretical guarantees including local convergence for the primal problem and convergence to a stationary point for the dual.
	
	Experiments on binary symmetric channels and modular sum queries validate the algorithms, showing improved privacy-utility tradeoffs over classical differential privacy mechanisms. This work provides a computational framework for auditing privacy risks and designing certified mechanisms under realistic adversary assumptions.
\end{abstract}

\section{Introduction}
\label{sec:section-introduction}

The exponential growth of data collection, driven by advances in artificial intelligence~\cite{2023arXiv230602781G} and big data analytics~\cite{mayer2013big}, has heightened the need for robust protections while preserving data utility. Sensitive personal data remains vulnerable to extraction and inference attacks, even from modern AI models~\cite{carlini2021extracting,hu2022membership}. This tension between utility and privacy is fundamentally captured by Dalenius' desideratum~\cite{dalenius1977towards,DBLP:conf/icalp/Dwork06}: rendering a mechanism's output can reveal minimal information about any individual record while remaining useful.

Differential Privacy (DP)~\cite{DBLP:journals/fttcs/DworkR14} has become a standard model by providing formal, composable guarantees. However, its (Dalenius) security relies on an often-unrealistic independence assumption among records, and its mechanisms can incur high utility losses. In the paper~\cite{WU2026123035}---henceforth referred to as the \textit{information privacy (IP) framework}---the authors addressed these limitations by introducing a bounded-knowledge adversary model via an entropy constraint $H(X) \geq b$. This replaces DP's independence assumption with the more realistic premise that an adversary has at least $b$ bits of uncertainty about the dataset.

Within the IP framework, we formalized three core optimization problems that directly parallel fundamental constructs in information theory:
\begin{enumerate}[label=\arabic*.]
	\item \textit{Maximal Per-Record Leakage}: Computing the worst-case information disclosure about any single record under the entropy constraint.
	\item \textit{Primal Leakage-Distortion Function}: Designing mechanisms that minimize this worst-case leakage subject to a utility (distortion) constraint $D$.
	\item \textit{Dual Minimal-Distortion Function}: The complementary problem of minimizing expected distortion under a strict leakage bound $L$.
\end{enumerate}

While~\cite{WU2026123035} established the theoretical framework and proved some key properties using information-theoretic approaches, it identified the \emph{computation} of these quantities as a significant open challenge. The problems are inherently high-dimensional and involve complex, entropy-constrained optimization over both adversarial priors and mechanism designs, making direct solution intractable.

\textbf{This paper directly bridges this computational gap.} We develop efficient, provably convergent algorithms to solve the three core problems of the IP framework. Our central insight is that, despite their apparent complexity, these problems possess a structure---specifically, the convexity-concavity duality of mutual information partially identified in~\cite{WU2026123035}---that is amenable to alternating optimization. Inspired by the classic Blahut-Arimoto algorithm for channel capacity computation~\cite{DBLP:journals/tit/Blahut72,DBLP:journals/tit/Arimoto72}, our methods exploit this duality to enable tractable and scalable computation.

\subsection{Our Main Contributions}
\begin{enumerate}[label=\arabic*.]
	\item \textit{Efficient Algorithms}: We develop novel alternating optimization algorithms for maximal leakage computation (Algorithm~\ref{alg:max_leakage}), the primal leakage-distortion tradeoff (Algorithm~\ref{alg:primal_tradeoff}), and the dual distortion minimization problem (Algorithm~\ref{alg:dual_tradeoff}).
	\item \textit{Theoretical Guarantees}: We prove local convergence for leakage computation and the primal problem, and convergence to a stationary point (satisfying KKT conditions) for the dual.
	\item \textit{Empirical Validation}: Extensive experiments on binary symmetric channels and modular sum queries achieve favorable privacy-utility tradeoffs compared to classical DP mechanisms under our framework's algorithmic. The results demonstrate that mechanisms optimized under our framework realize the utility benefits promised by the entropy-constrained model.
\end{enumerate}

This work thus provides a \emph{computational foundation} for the IP framework, enabling both practical privacy risk auditing of privacy mechanisms (by computing exact leakages), including the differentially private mechanisms, and the design of certified privacy mechanisms under realistic adversary assumption.

\section{Problem Formalization} \label{sec:problem-formalization}

\subsection{The Information Privacy Framework: A Recap}
\label{subsec:information-privacy-framework}

\begin{table}[!tbh]
	\caption{Nomenclature}
	\label{tab:nomenclature}
	\footnotesize
	\centering
	\begin{tabular}{@{} p{0.22\textwidth} p{0.76\textwidth} @{}}
		\toprule
		\textbf{Symbol/Term} & \textbf{Description} \\
		\midrule
		$x=(x_1,\ldots,x_n)$ &  The true secret dataset (a realization of $X$). \\
		$X = (X_1, \ldots, X_n)$ & Random vector representing an adversary's prior belief over the dataset. \\
		$n$ & The number of records (tuples) in a dataset. \\
		$[m]$ & The set $\{1,2,\ldots,m\}$ \\
		$\mathcal{X}_i =\{x_{i1},\ldots,x_{in_i}\}$ & Alphabet (universe) of the $i$-th record \\
		$n_i, x_i$ & The size of $\mathcal X_i$, and an element within $\mathcal X_i$\\
		$\mathcal{X} = \prod_{i=1}^n \mathcal{X}_i$ & Universe of all possible datasets \\
		$\mathcal{X}_{-i}$, $n_{-i}$, $x_{-i}$ & Alphabet of all records except the $i$-th, its size, and an element within it. \\
		$\mathcal{Y}$ & Output alphabet of the privacy mechanism. \\
		$e_j$ & Unit vector $(0,\ldots,0,1,0,\ldots,0)$ with 1 in the $j$-th position. \\
		$p(x)$ & Joint probability distribution of the adversary's belief $X$. \\
		$p(x_i)$ & Marginal probability distribution of the $i$-th record $X_i$. \\
		$p(x_{-i}|x_i)$ & Conditional probability distribution (transition matrix) of $X_{-i}$ given $X_i$. \\
		$\mathcal E$ &  The set of all unit vectors in $\mathbb R^{n_{-i}}$ \\
		$p_0(x)$ & The true underlying distribution from which datasets are generated. \\
		$q(y|x)$ & The privacy mechanism, a conditional distribution mapping dataset $x$ to output $y$. \\
		$Y$ & The random variable representing the output of the privacy mechanism. \\
		$f(x)$ & The original query function, $f: \mathcal{X} \to \mathcal{Y}$. \\
		$d(\cdot, \cdot)$ & A distortion metric on the output space $\mathcal{Y} \times \mathcal{Y}$. \\
		$D$ & The maximum allowed expected distortion (a constraint). \\
		$I(X_i; Y)$ & Mutual information between record $X_i$ and mechanism output $Y$. \\
		$H(X)$ & Shannon entropy of the joint distribution $p(x)$. \\
		$H(X_{-i}|x_i)$ & Conditional entropy of $X_{-i}$ given $X_i = x_i$. \\
		$b$ & The minimum entropy constraint ($H(X) \geq b$) on the adversary's prior. \\
		$\mathcal{L}(b)$ & Maximal per-record leakage under entropy constraint $b$ (Def. 2.1). \\
		$\epsilon(D, b)$ & Primal leakage-distortion function: min max leakage given distortion $D$ (Def. 2.2). \\
		$\mathcal D(L, b)$ & Dual minimal-distortion function: min distortion given max leakage $L$ (Def. 2.3). \\
		$L$ & The maximum allowed per-record leakage (a constraint). \\
		$\Delta(\mathcal{X})$ & The probability simplex over the alphabet $\mathcal{X}$. \\
		$\Delta(\mathcal{X})_b$ & The subset of $\Delta(\mathcal{X})$ where distributions satisfy $H(X) \geq b$. \\
		$\mathbb{P}_{\text{ind}}$ & The set of all product distributions on $\mathcal{X}$ (DP's assumed adversary set). \\
		$\mathbb E_{p_0}[d] = \mathbb E_{p_0} d(f(X),Y)$ & Expected distortion $\sum_{x\in\mathcal X, y\in\mathcal Y} p_0(x)q(y|x)d(f(x),y)$ \\
		\bottomrule
	\end{tabular}
\end{table}

Our work builds directly upon the \textbf{information privacy (IP) framework} introduced in~\cite{WU2026123035}. That work established a novel, information-theoretic foundation for Dalenius' privacy vision by replacing the often-unrealistic independence assumption of differential privacy with a more practical \textbf{bounded-knowledge adversary model}. This model is formalized via an entropy constraint on the adversary's prior belief.

We briefly recap the core components of this framework, as they form the foundation for the computational problems addressed in this paper.

\begin{itemize}
	\item \textbf{Dataset and Adversary Model:} A sensitive dataset contains $n$ records, $x = (x_1, \ldots, x_n)$. An adversary's prior knowledge is modeled as a random vector $X = (X_1, \ldots, X_n)$ with a joint probability distribution $p(x)$ over the universe $\mathcal{X} = \prod_{i=1}^n \mathcal{X}_i$.
	
	\item \textbf{Privacy Mechanism:} A privacy mechanism for a query function $f: \mathcal{X} \to \mathcal{Y}$ is modeled as a stochastic channel $q(y|x)$, which releases a randomized output $Y$ instead of the true result $f(x)$.
	
	\item \textbf{Privacy Metric:} The information leakage about an individual record $x_i$ is quantified by the \textbf{Shannon mutual information} $I(X_i; Y)$. This measures the average reduction in uncertainty about $X_i$ after observing the mechanism's output $Y$.
	
	\item \textbf{Core Assumption (Bounded Adversary Knowledge):} The key innovation of the IP framework is to replace Differential Privacy's independence assumption with a more realistic \textbf{entropy constraint} $H(X) \geq b$, where $H(X)$ is the Shannon entropy of the adversary's prior $p(x)$. This formalizes the premise that in large-scale datasets, an adversary must have at least $b$ bits of uncertainty about the full dataset---a more tenable assumption than requiring complete statistical independence between all records.
	
	\item \textbf{Dalenius Security in the IP Framework:} A mechanism $q(y|x)$ satisfies $\epsilon$-\textbf{Dalenius security} with respect to an adversary class $\mathbb{P} \subseteq \Delta(\mathcal{X})$ if:
	\begin{align}
		\max_{i \in [n]} \max_{p(x) \in \mathbb{P}} I(X_i; Y) \leq \epsilon.
	\end{align}
	When $\mathbb{P} = \Delta(\mathcal{X})_b = \{p(x): H(X) \geq b\}$, this defines security against all \textbf{entropy-constrained adversaries}.
\end{itemize}

The paper~\cite{WU2026123035} proved that fundamental mechanisms (randomized response, exponential, Gaussian) satisfy this notion of security, established group and composition properties, and demonstrated superior utility-privacy tradeoffs compared to classical DP. However, a central challenge was left open: the \textbf{practical computation} of key quantities like the maximal per-record leakage under the entropy constraint, which is essential for auditing mechanisms and optimizing the privacy-utility tradeoff.

\subsection{Problem Statements: Bridging Theory and Computation}
\label{subsec:problem-statements}

This paper directly addresses the computational gap identified in~\cite{WU2026123035}. We formalize and solve the core optimization problems that are necessary to operationalize the framework. These problems naturally extend the information-theoretic constructs of the IP model into the domain of algorithmic mechanism analysis and design.

\begin{definition}[Maximal Per-Record Leakage]
	\label{def:max-leakage}
	For a fixed privacy mechanism $q(y|x)$, compute the worst-case information leakage about any individual record against an entropy-constrained adversary:
	\begin{align}
		\mathcal{L}(b) \triangleq \max_{i \in [n]} \max_{\substack{p(x) \in \Delta(\mathcal{X}) \\ H(X) \geq b}} I(X_i; Y).
	\end{align}
	This quantity, $\mathcal{L}(b)$, is the \textbf{individual channel capacity} relative to $\Delta(\mathcal{X})_b$. Computing it is the essential first step for quantifying the privacy risk of any given mechanism under the bounded-knowledge model.
\end{definition}

Building on this, we formulate the central optimization problems for mechanism design:

\begin{definition}[Primal Leakage-Distortion Tradeoff]
	\label{def:primal-tradeoff}
	Given a utility requirement specified as a maximum allowable expected distortion $D \geq 0$, find the mechanism that minimizes the worst-case leakage:
	\begin{align}
		\epsilon(D, b) \triangleq \min_{q(y|x)} \max_{i \in [n]} \max_{\substack{p(x) \in \Delta(\mathcal{X}) \\ H(X) \geq b}} I(X_i; Y) \quad \text{s.t.} \quad \mathbb{E}_{p_0}[d(f(X), Y)] \leq D.
	\end{align}
	Here, $p_0(x)$ is the true data distribution, and $d(\cdot, \cdot)$ is a distortion metric. This is the \textbf{primal problem}: minimizing privacy loss subject to a utility constraint.
\end{definition}

\begin{definition}[Dual Minimal-Distortion Formulation]
	\label{def:dual-tradeoff}
	Given a strict privacy budget $L \geq 0$ (maximal allowable leakage), find the mechanism that minimizes distortion:
	\begin{align}
		\mathcal D(L, b) \triangleq \min_{q(y|x)} \mathbb{E}_{p_0}[d(f(X), Y)] \quad \text{s.t.} \quad \max_{i \in [n]} \max_{\substack{p(x) \in \Delta(\mathcal{X}) \\ H(X) \geq b}} I(X_i; Y) \leq L.
	\end{align}
	This is the \textbf{dual problem}: minimizing utility loss subject to a privacy constraint.
\end{definition}

Problems in Definitions~\ref{def:primal-tradeoff} and~\ref{def:dual-tradeoff} form a \textbf{complementary pair}, whose solutions trace the Pareto-optimal frontier between privacy (leakage) and utility (distortion) in the IP framework. Solving these problems provides both a method to \textbf{audit} the leakage of existing mechanisms (like DP mechanisms) under the entropy-constrained model and to \textbf{design} certified optimal mechanisms for this more realistic adversary.

\subsection{Structural Properties: Theoretical Foundations for Computation}
\label{subsec:structural-properties}

The computational challenges posed by Problems in Definitions~\ref{def:max-leakage}--\ref{def:dual-tradeoff} are formidable due to their high-dimensional, constrained nature. However, they possess crucial structural properties---derived from information theory---that enable efficient algorithmic solutions. These properties form the mathematical backbone of our approach.

\textbf{Convexity-Concavity of Mutual Information.} The mutual information \(I(X_i; Y)\), central to all three problems, exhibits a dual convex-concave structure with respect to its parameters:

	\begin{theorem}[Mutual Information Convexity-Concavity Duality]
	\label{thm:mutual_info_properties}
	The mutual information $I(X_i; Y)$ exhibits the following structural properties:
	\begin{enumerate}
		\item \textbf{Concavity in marginals}: For fixed $p(x_{-i}|x_i)$ and $q(y|x)$, 
		$I(X_i; Y)$ is concave in $p(x_i)$.
		
		\item \textbf{Convexity in conditionals}: For fixed $p(x_i)$ and $q(y|x)$, 
		$I(X_i; Y)$ is convex in $p(x_{-i}|x_i)$.
		
		\item \textbf{Convexity in mechanisms}: For fixed $p(x) = p(x_i)p(x_{-i}|x_i)$, 
		$I(X_i; Y)$ is convex in $q(y|x)$.
	\end{enumerate}
\end{theorem}

\begin{proof}
	\textbf{Proof of (1):} Fix $p(x_{-i}|x_i)$ and $q(y|x)$. The marginal channel $p(y|x_i)$ is fixed. Mutual information can be expressed as:
	$$
	I(X_i; Y) = D\left( p(y|x_i) p(x_i) \| p(x_i) p(y) \right),
	$$
	which is concave in $p(x_i)$ for fixed $p(y|x_i)$ by Theorem 2.7.4 in \cite{DBLP:books/daglib/0016881}.
	
	\textbf{Proof of (2):} Fix $p(x_i)$ and $q(y|x)$. For $\lambda \in [0,1]$, consider two distributions $p_1(x_{-i}|x_i)$, $p_2(x_{-i}|x_i)$ with convex combination:
	$$
	p_\lambda(x_{-i}|x_i) = \lambda p_1(x_{-i}|x_i) + (1-\lambda) p_2(x_{-i}|x_i).
	$$
	The joint distribution is linear in the conditionals:
	$$
	p_\lambda(x_i, y) = p(x_i) \sum_{x_{-i}} p_\lambda(x_{-i}|x_i) q(y|x_i, x_{-i}) = \lambda p_1(x_i, y) + (1-\lambda) p_2(x_i, y).
	$$
	The output marginal becomes:
	$$
	p_\lambda(y) = \sum_{x_i} p_\lambda(x_i, y) = \lambda p_1(y) + (1-\lambda) p_2(y).
	$$
	Defining reference distributions $q_k(x_i, y) = p(x_i) p_k(y)$, we have:
	$$
	q_\lambda(x_i, y) = p(x_i) p_\lambda(y) = \lambda q_1(x_i, y) + (1-\lambda) q_2(x_i, y).
	$$
	Mutual information equals the KL divergence:
	$$
	I_\lambda(X_i; Y) = D(p_\lambda(x_i, y) \| q_\lambda(x_i, y)) \leq \lambda D(p_1 \| q_1) + (1-\lambda) D(p_2 \| q_2)
	$$
	by joint convexity of KL divergence \cite{DBLP:books/daglib/0016881}. Thus, $I(X_i; Y)$ is convex in $p(x_{-i}|x_i)$.
	
	\textbf{Proof of (3):} Fix $p(x) = p(x_i)p(x_{-i}|x_i)$. For $\lambda \in [0,1]$, consider two mechanisms $q_A(y|x)$, $q_B(y|x)$ with convex combination:
	$$
	q_\lambda(y|x) = \lambda q_A(y|x) + (1-\lambda) q_B(y|x).
	$$
	The joint distribution is linear:
	$$
	p_\lambda(x, y) = p(x) q_\lambda(y|x) = \lambda p_A(x, y) + (1-\lambda) p_B(x, y).
	$$
	Since $I(X_i; Y) = \sum_{x_i} p(x_i) D(p(y|x_i) \| p(y))$, and both $p(y|x_i)$ and $p(y)$ are linear functions of $q(y|x)$, mutual information is convex in $q(y|x)$ by preservation of convexity through linear combinations and KL divergence \cite{DBLP:books/daglib/0016881}.
\end{proof}

Theorem~\ref{thm:mutual_info_properties} is motivated by Theorem~2.7.4 in~\cite{DBLP:books/daglib/0016881}. 
This convex--concave duality enables efficient alternating optimization strategies, where one set of variables can be optimized while keeping the others fixed. 
Notably, the \textbf{convexity in conditionals} applies to each row of the transition matrix \(p(x_{-i}|x_i)\) individually, rather than to the matrix as a whole. 
In this paper, we primarily leverage row-wise convexity, as detailed in Section~\ref{subsubsec:conditionals_update}.

\textbf{Convexity of Constraint Sets.} The entropy constraint \(H(X) \geq b\) and leakage constraints define convex feasible regions:

\begin{theorem}[Convexity of Constraint Sets]  \label{thm:entropy_concavity}
	The following sets are convex:
	\begin{enumerate}
		\item The set of joint distributions with entropy at least \( b \): 
		\begin{align}
			\bigl\{ p(x) \in \Delta(\mathcal{X}) : H(X) \geq b \bigr\}.
		\end{align}
		\item For a fixed marginal distribution \( p(x_i) > 0 \), the set of conditional distributions satisfying the entropy constraint is convex: 
		\begin{align}
			\bigl\{ p(x_{-i} | x_i) : H(X) \geq b \bigr\}.
		\end{align}
		Similarly, for fixed conditionals \( p(x_{-i} | x_i) \), the set of marginals \( \bigl\{ p(x_i) : H(X) \geq b \bigr\} \) is convex.
		\item For a fixed joint distribution \( p(x) \), the set of mechanisms that achieve at most \( L \) leakage on record \( i \) is convex: 
		\begin{align}
			\bigl\{ q(y | x) : I(X_i; Y) \leq L \bigr\}.
		\end{align}
	\end{enumerate}
\end{theorem}

\begin{proof}
	\begin{enumerate}
		\item[(1)] The function \( H(X) \) is concave in \( p(x) \)~\cite{DBLP:books/daglib/0016881}. The set \( \bigl\{ p(x) : H(X) \geq b \bigr\} \) is a super-level set of a concave function and is therefore convex \cite[Section 3.1]{citeulike:163662}.
		
		\item[(2)] For a fixed marginal \( p(x_i) \), the joint entropy can be expressed as \( H(X) = H(X_i) + \mathbb{E}\bigl[ H(X_{-i} | X_i) \bigr] \). Since \( H(X_i) \) is fixed, the constraint \( H(X) \geq b \) is equivalent to \( \mathbb{E}\bigl[ H(X_{-i} | X_i) \bigr] \geq b - H(X_i) \). The function \( H(X_{-i} | x_i) \) is concave in \( p(x_{-i} | x_i) \) for each \( x_i \), and the expectation over \( x_i \) preserves concavity. Thus, the constraint is a super-level set of a concave function of the conditionals \( p(x_{-i} | x_i) \) and is convex. The logic for a fixed conditional is analogous, as \( H(X) \) is concave in \( p(x_i) \).
		
		\item[(3)] For a fixed \( p(x) \), Theorem \ref{thm:mutual_info_properties} establishes that \( I(X_i; Y) \) is convex in the mechanism \( q(y | x) \). The set \( \bigl\{ q(y | x) : I(X_i; Y) \leq L \bigr\} \) is a sub-level set of a convex function and is therefore convex \cite[Section 3.1]{citeulike:163662}.
	\end{enumerate}
\end{proof}

\begin{lemma}[Representation of Entropy Super-Level Sets]  \label{lemma:convex-representation-entropy-super-level-set}
	Let \( \mathcal{X} \) be a finite set and \( \Delta(\mathcal{X}) \) the probability simplex over \( \mathcal{X} \). Denote the Shannon entropy by \( H \). For any constant \( b \in [0,\,\log |\mathcal{X}|] \), define
	\[
	S = \{ p \in \Delta(\mathcal{X}) : H(p) \geq b \}, \qquad 
	\partial S = \{ p \in \Delta(\mathcal{X}) : H(p) = b \}.
	\]
	Then \( S = \operatorname{conv}(\partial S) \), where \( \operatorname{conv} \) denotes the convex hull. Consequently, every distribution in \( S \) can be expressed as a finite convex combination of distributions in \( \partial S \).
\end{lemma}

\begin{proof}
	The proof follows these four steps.
	
	\textbf{Step 1: Preliminaries.}
	The probability simplex \( \Delta(\mathcal{X}) \) is compact and convex. The entropy function \( H \) is continuous on \( \Delta(\mathcal{X}) \) and strictly concave because \( t \mapsto -t\log t \) is strictly concave on \( [0,\infty) \) and \( H \) is a sum of such functions.
	
	\textbf{Step 2: Properties of \( S \) and \( \partial S \).}
	Since \( H \) is continuous, \( S \) is closed as the preimage of the closed interval \( [b, \log|\mathcal{X}|] \). As a closed subset of the compact set \( \Delta(\mathcal{X}) \), \( S \) is itself compact. The strict concavity of \( H \) implies that its super-level set \( S \) is convex. The set \( \partial S \) is nonempty because a continuous function on a connected, compact set attains all intermediate values.
	
	\textbf{Step 3: Extreme points of \( S \) lie in \( \partial S \).}
	Let \( p \) be an extreme point of \( S \). Suppose, for contradiction, that \( H(p) > b \).
	By continuity of \( H \), there exists \( \epsilon > 0 \) such that
	\begin{align}
		B_{\epsilon}(p) \cap \operatorname{aff}\bigl(\Delta(\mathcal{X})\bigr) \subseteq S, 
	\end{align}
	where \( \operatorname{aff} \) denotes the affine hull.
	We now construct two distinct points \( r, s \) in this neighbourhood whose midpoint is \( p \).
	Since \( H(p) > b \geq 0 \), the distribution \( p \) cannot be a vertex of the simplex (which would have entropy \( 0 \)). Therefore, \( p \) has at least two coordinates that are strictly positive. Denote two such coordinates by \( i \) and \( j \) (\( p_i > 0, p_j > 0 \)).
	Define the direction vector \( d \in \mathbb{R}^{|\mathcal{X}|} \) by
	\[
	d_i = 1, \qquad d_j = -1, \qquad d_k = 0 \;\; \text{for all } k \notin \{i, j\}.
	\]
	Then \( \sum_x d_x = 0 \), so \( d \) lies in the linear subspace parallel to \( \operatorname{aff}(\Delta(\mathcal{X})) \).
	Choose a step size \( \delta > 0 \) satisfying the following three conditions:
	\begin{enumerate}
		\item \( \delta \le \epsilon / \| d \| \) (so that \( p \pm \delta d \in B_\epsilon(p) \)),
		\item \( \delta \le p_i \) and \( \delta \le p_j \) (guaranteeing non-negativity of the components of \( p \pm \delta d \)),
		\item \( \delta \) is small enough that \( H(p \pm \delta d) \ge b \) (by continuity of \( H \) and \( H(p) > b \)).
	\end{enumerate}
	Now set \( r = p + \delta d \) and \( s = p - \delta d \). By construction,
	\( r, s \in \Delta(\mathcal{X}) \cap B_\epsilon(p) \subseteq S \), and \( p = \frac{1}{2}(r + s) \) with \( r \neq s \).
	This contradicts the hypothesis that \( p \) is an extreme point of \( S \). Hence, \( H(p) = b \) and \( p \in \partial S \).
	
	\textbf{Step 4: Krein-Milman theorem and convex hull.}
	The set \( S \) is a nonempty compact convex set in \( \mathbb{R}^{|\mathcal{X}|} \). By the Krein-Milman theorem, \( S \) is the closed convex hull of its extreme points:
	\begin{align}
		S = \overline{\operatorname{conv}}\bigl(\operatorname{ext}(S)\bigr),
	\end{align}
	where $\operatorname{ext}(S)$ is the set of extreme points.
	From Step~3, we have \( \operatorname{ext}(S) \subseteq \partial S \). Consequently,
	\[
	S \subseteq \overline{\operatorname{conv}}(\partial S).
	\]
	The set \( \partial S \) is closed (as the preimage of the singleton \( \{b\} \) under the continuous function \( H \)) and, being a subset of the compact set \( \Delta(\mathcal{X}) \), is itself compact. In finite dimensions, the convex hull of a compact set is compact and therefore closed; hence \( \overline{\operatorname{conv}}(\partial S) = \operatorname{conv}(\partial S) \). Finally, because \( \partial S \subseteq S \) and \( S \) is convex, we have \( \operatorname{conv}(\partial S) \subseteq S \). Combining the inclusions yields
	\[
	S = \operatorname{conv}(\partial S).
	\]
	Carathéodory's theorem ensures that every point in \( \operatorname{conv}(\partial S) \) can be expressed as a convex combination of a finite number of points from \( \partial S \), which completes the proof.
\end{proof}

These structural properties---convexity-concavity duality, constraint convexity, entropy decomposition, and permutation invariance---provide the mathematical foundation for the efficient algorithms developed in Section~\ref{sec:algorithms}. They transform seemingly intractable high-dimensional optimization problems into sequences of manageable convex subproblems.

\section{Algorithms for Maximal Leakage and Leakage-Distortion Tradeoffs}
\label{sec:algorithms}

This section presents algorithms for solving the three problems defined in Section \ref{sec:problem-formalization}: maximal per-record leakage (Definition \ref{def:max-leakage}), the primal leakage-distortion trade-off (Definition \ref{def:primal-tradeoff}), and the dual minimal-distortion formulation (Definition \ref{def:dual-tradeoff}). The core of our approach is an \textbf{alternating optimization framework} that leverages the structural properties---specifically the convexity and concavity of mutual information and entropy---established in Section \ref{subsec:structural-properties}.

The algorithmic strategies are as follows:
\begin{itemize}
	\item For \textbf{Problem 1} (Definition \ref{def:max-leakage}), we alternate between optimizing the marginal distribution \( p(x_i) \) and the conditional distributions \( p(x_{-i} | x_i) \).
	\item For \textbf{Problem 2} (Definition \ref{def:primal-tradeoff}), we alternate between optimizing the adversarial prior \( p(x) \) and the privacy mechanism \( q(y | x) \).
	\item For \textbf{Problem 3} (Definition \ref{def:dual-tradeoff}), we exploit the convexity of the mechanism set under a leakage constraint to achieve global convergence.
\end{itemize}
All algorithms explicitly enforce the entropy constraint \( H(X) \geq b \) during updates.

\subsection{Algorithm for Maximal Per-Record Leakage (Problem 1)}
\label{subsec:max_leakage_alg}

We now address the optimization problem in Definition \ref{def:max-leakage}:
\begin{subequations} \label{eq:optm-max_leakage}
	\begin{align} 
		\mathcal{L}(b) \triangleq & \quad \max_{i \in [n], p(x)}  I(X_i;Y) \\
		s.t. & \quad H(X) \geq b 
		\label{eq:equation-4}
	\end{align}
\end{subequations}
where \( p(x) = p(x_i) p(x_{-i} | x_i) \), and the privacy mechanism \( q(y| x) \) is fixed. The entropy constraint $H(X) \geq b$ can be expanded using the chain rule of entropy:
\begin{align} \label{eq:entropy-constraint-expantion}
	H(X) = H(X_i) + \sum_{x_i} p(x_i) H(X_{-i}|x_i) \ge b.
\end{align}
There are two distinct cases to the constraint:
\begin{enumerate}
	\item \textbf{Case 1: $b \leq H(X_i)$}. In this scenario, the inequality \eqref{eq:entropy-constraint-expantion} is automatically satisfied since $H(X_{-i}|x_i)\ge 0$
	 for all possible conditional distributions $p(x_{-i}|x_i)$. The optimization problem \eqref{eq:optm-max_leakage} thus becomes unconstrained. Then $I(X_i;Y)$ obtains its maximum when each row of the transition matrix $p(x_{-i}|x_i)$ is in $\mathcal E$, the set of all unit vectors in $\mathbb R^{n_{-i}}$, due to concavity of $I(X_i;Y)$ to $p(x_{-i}|x_i)$ (Theorem \ref{thm:mutual_info_properties}) and the fact that each row of $p(x_{-i}|x_i)$ must be a convex combination of elements in $\mathcal E$. 
	This aligns with the intuition that to maximize the information leakage from a single record, the rest of the dataset must be perfectly predictable given that record, minimizing uncertainty.
	\item \textbf{Case 2: $b > H(X_i)$}. This is the non-trivial case where $\sum_{x_i} p(x_i) H(X_{-i}|x_i)>b-H(X_i)>0$. The conditional distributions $p(x_{-i}|x_i)$ must contribute enough uncertainty to satisfy the global constraint $H(X) \geq b$, imposing a complex restriction on the domain of the probability transition matrix $p(x_{-i}|x_i)$.
\end{enumerate}

The remainder of this section focuses on Case 2 ($b > H(X_i)$), as Case 1 admits a direct solution. The foundation of our algorithm is the convex-concave structure of the objective function $I(X_i;Y)$ with respect to its parameters, as established by Theorem \ref{thm:mutual_info_properties}:
\begin{itemize}
	\item For fixed conditionals $p(x_{-i}|x_i)$, $I(X_i;Y)$ is \textbf{concave} in the marginal distribution $p(x_i)$.
	\item For fixed marginal $p(x_i)$, $I(X_i;Y)$ is \textbf{convex} in the conditionals $p(x_{-i}|x_i)$.
\end{itemize}
Figure~\ref{fig:saddle-surface} illustrates the convex-concave structure of mutual information $I(X_i;Y)$.

\begin{algorithm}[!tbh]
	\caption{Alternating Optimization for Maximal Per-Record Leakage}
	\label{alg:max_leakage}
	\begin{algorithmic}[1]
		\REQUIRE Fixed mechanism $q(y|x)$, entropy constraint $b > 0$, tolerance $\epsilon > 0$, maximum iterations $T_{\max}$
		\ENSURE Optimal joint distribution $p^{*}(x)$, maximal leakage $\mathcal{L}(b)$, optimal record $i_{\max}$ 
		\STATE Randomly initialize $p^{(0)}(x)$ such that $H(p^{(0)}(x))\ge b$ \COMMENT{Ensure feasibility}
		\STATE $\mathcal{L}_{\text{prev}} \gets -\infty$, $\mathcal{L}_{\max} \gets -\infty$, $p_{\max}(x) \gets \emptyset$, $i_{\max} \gets -1$
		\STATE $\mathcal{E} \gets$ the set of all unit vectors in $\mathbb{R}^{n_{-i}}$
		\STATE $k \gets 0$, $\Delta\mathcal{L} \gets \infty$
		
		\WHILE{$\Delta\mathcal{L} \geq \epsilon$ AND $k < T_{\max}$}
			\STATE $\mathcal{L}_{\text{prev}} \gets \mathcal{L}_{\max}$  
			
			\FOR{$i \in [n]$}
				\STATE Extract $p(x_i)$ and $p(x_{-i}|x_i)$ from $p^{(k)}(x)$
				
				\STATE \textbf{Step 1: Update marginal distribution}
				\STATE $p^{\text{new}}(x_i) \gets \text{Algorithm \ref{alg:marginal_optimization}}(p(x_i), p(x_{-i}|x_i), q, b)$
				
				\STATE \textbf{Step 2: Update conditional distributions}
				\IF{$b > H(p^{\text{new}}(x_i))$}
					\STATE $p^{\text{new}}(x_{-i}|x_i) \gets \text{Algorithm \ref{alg:conditional_optimization_gradient_entropy_manifold}}(p^{\text{new}}(x_i), p(x_{-i}|x_i), q, b)$
				\ELSE
					\STATE $p^{\text{new}}(x_{-i}|x_i) \gets \arg\max_{p' \in \mathcal{E}^{n_i}} I(X_i;Y)$ \COMMENT{Search over deterministic conditionals}
				\ENDIF
				
				\STATE Construct $p^{\text{candidate}}(x) = p^{\text{new}}(x_i) \cdot p^{\text{new}}(x_{-i}|x_i)$
				\STATE $I_{\text{current}} \gets I(X_i;Y)$ using $p^{\text{candidate}}(x)$ and $q(y|x)$
				
				\IF{$I_{\text{current}} > \mathcal{L}_{\max}$}
					\STATE $\mathcal{L}_{\max} \gets I_{\text{current}}$
					\STATE $p_{\max}(x) \gets p^{\text{candidate}}(x)$
					\STATE $i_{\max} \gets i$
				\ENDIF
			\ENDFOR
			
			\STATE $p^{(k+1)}(x) \gets p_{\max}(x)$
			\STATE $\Delta\mathcal{L} \gets |\mathcal{L}_{\max} - \mathcal{L}_{\text{prev}}|$
			\STATE $k \gets k + 1$
		\ENDWHILE
		
		\STATE \textbf{return} $p^{*}(x) \gets p^{(k)}(x)$, $\mathcal{L}(b) \gets \mathcal{L}_{\max}$, $i_{\max}$
	\end{algorithmic}
\end{algorithm}

This structure naturally motivates an \textbf{alternating optimization approach}. We iteratively optimize one set of variables while holding the other fixed, guaranteeing that each step either increases the objective function or leaves it unchanged. This procedure converges to limit values of the objective function.

Algorithm \ref{alg:max_leakage} outlines this iterative procedure. It initializes the joint distribution $p(x)$ to the uniform distribution, which maximizes entropy and ensures feasibility. The main loop alternates between, optimizing the marginal distribution $p(x_i)$ (using Algorithm \ref{alg:marginal_optimization}), and then optimizing its corresponding conditional distributions $p(x_{-i}|x_i)$ (using Algorithm \ref{alg:conditional_optimization_gradient_entropy_manifold} for Case 2). If the two sub-algorithms fail to improve the objective for all record $i\in [n]$, the update will be rejected. The algorithm terminates when the change in mutual information falls below a specified tolerance $\epsilon$.

\begin{figure}[!htbp]
	\centering
	\includegraphics[width=1\textwidth]{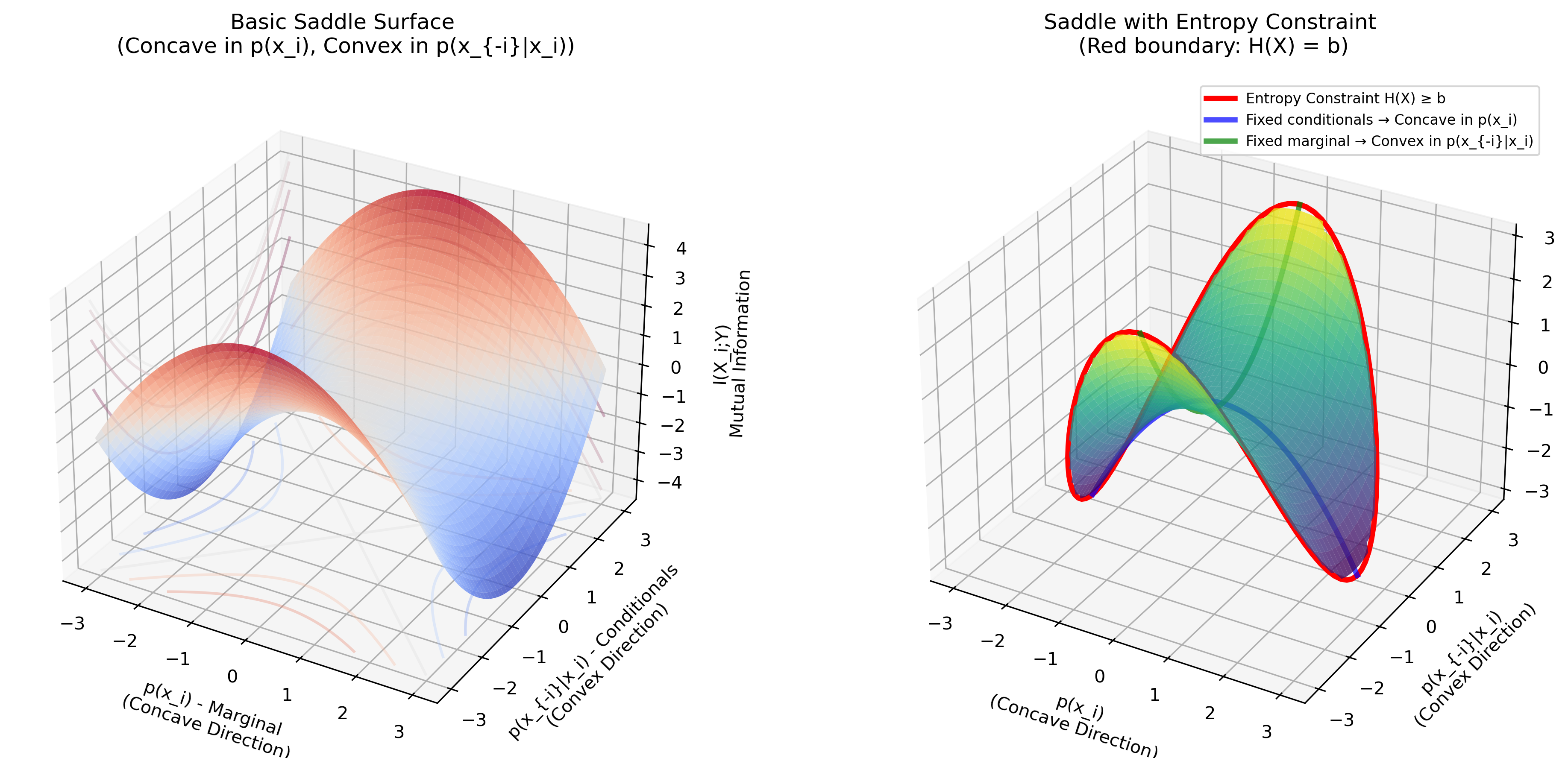}
	\caption{
		\textbf{Saddle-shaped structure of mutual information} under convex-concave duality. 
		(\textbf{Left}) Basic saddle surface demonstrating the convex-concave duality of mutual information $I(X_i;Y)$: 
		concave in the marginal distribution $p(x_i)$ and convex in the conditional distribution $p(x_{-i}|x_i)$. 
		(\textbf{Right}) Saddle surface with entropy constraint $H(X) \geq b$ (red boundary), illustrating the feasible region 
		for entropy-constrained optimization. Blue and green curves demonstrate cross-sections showing the concave/convex 
		properties along fixed conditionals and marginals, respectively.
	}
	\label{fig:saddle-surface}
\end{figure}

\subsubsection{Optimal Marginal \( p^*(x_i) \) for Fixed Transition Matrix \( p(x_{-i}|x_i) \)}  
\label{subsubsec:marginal_update}

Given a fixed probability transition matrix \( p(x_{-i}|x_i) \) and privacy mechanism \( q(y|x) \), the channel \( p(y|x_i) \) is determined as:
\begin{align}
	p(y|x_i) = \sum_{x_{-i}} p(x_{-i}|x_i) q(y|x_i, x_{-i}).
\end{align}
The optimization problem (\ref{eq:optm-max_leakage}) reduces to:
\begin{subequations}
	\label{eq:optm-marginal}
	\begin{align}
		\max_{p(x_i)} \quad & I(X_i; Y) = I(p(x_i), p(y|x_i)) \\
		\text{s.t.} \quad & H(X) \geq b \\
		& p(x_i) \geq 0, \quad \sum_{x_i} p(x_i) = 1,
	\end{align}
\end{subequations}
where the joint entropy is:
\[
H(X) = H(X_i) + \sum_{x_i} p(x_i) H(X_{-i}|x_i).
\]
Here, \( H(X_{-i}|x_i) \) is constant for each \( x_i \) since \( p(x_{-i}|x_i) \) is fixed. By Theorems \ref{thm:mutual_info_properties} and \ref{thm:entropy_concavity}, the objective \( I(X_i; Y) \) is concave in \( p(x_i) \), and the constraint \( H(X) \geq b \) is concave, making Equation \eqref{eq:optm-marginal} a convex optimization problem.

\paragraph{Lagrangian Formulation}
We introduce the Lagrangian:
\[
\mathcal{L} = I(X_i; Y) + s \bigl(H(X) - b\bigr),
\]
where \( s \geq 0 \) is a Lagrange multiplier. The partial derivative with respect to \( p(x_i) \) is:
\[
\frac{\partial \mathcal{L}}{\partial p(x_i)} = \frac{\partial I(X_i; Y)}{\partial p(x_i)} + s \frac{\partial H(X)}{\partial p(x_i)}.
\]
Using foundational derivatives from Appendix \ref{sec:foundational_derivatives}, we obtain:
\begin{align}
	\frac{\partial \mathcal{L}}{\partial p(x_i)} = \sum_y p(y|x_i) \log \frac{p(y|x_i)}{p(y)} - 1 + s \left( H(X_{-i}|x_i) - \log p(x_i) - 1 \right).
\end{align}

\paragraph{Modified Blahut-Arimoto Update}
We derive a Blahut-Arimoto style update~\cite{DBLP:journals/tit/Blahut72,DBLP:journals/tit/Arimoto72} by considering the Lagrangian optimization. The update rule becomes:
\begin{align}
	p^{(t+1)}(x_i) \propto p^{(t)}(x_i) \exp \left( \frac{1}{1+s^{(t)}} \left[ \sum_y p(y|x_i) \log \frac{p(y|x_i)}{p^{(t+1)}(y)} \right] + \frac{s^{(t)}}{1+s^{(t)}} \left( H(X_{-i}|x_i) - \log p(x_i)  \right) \right),
\end{align}
where \( p^{(t+1)}(y) = \sum_{x_i} p^{(t)}(x_i) p(y|x_i) \) is the output marginal distribution.

The multiplier \( s \) is updated adaptively by the gradient descent method to enforce the entropy constraint:
\[
s^{(t+1)} = \max\left(0, s^{(t)} - \alpha (H(X^{(t+1)}) - b)\right),
\]
with \( \alpha > 0 \) (e.g., \( \alpha = 0.1 \)).

\paragraph{Algorithm Description}
Algorithm \ref{alg:marginal_optimization} implements this modified Blahut-Arimoto procedure~\cite{DBLP:journals/tit/Blahut72,DBLP:journals/tit/Arimoto72}. It alternates between computing the output distribution (E-step) and updating the input distribution (M-step) while adapting the Lagrange multiplier. The algorithm terminates when the change in mutual information falls below tolerance \( \epsilon \).

\begin{algorithm}[!bth]
	\caption{Modified Blahut-Arimoto for Marginal Optimization with \( H(X) \geq b \)}
	\label{alg:marginal_optimization}
	\begin{algorithmic}[1]
		\REQUIRE Fixed transition matrix \( p(x_{-i}|x_i) \), channel \( q(y|x) \), \( p^{(0)}(x_i) \), entropy constraint \( b > 0 \), tolerance \( \epsilon > 0 \), step size \( \alpha > 0 \)
		\ENSURE Optimal marginal distribution \( p^*(x_i) \), maximal mutual information \( I^* \)
		\STATE \( s^{(0)} \gets 0 \), \( t \gets 0 \)
		\STATE Precompute \( H(X_{-i}|x_i) \) and \( p(y|x_i) \) for all \( x_i \)
		\REPEAT
		\STATE \textbf{E-step: Compute output distribution}
		\FOR{each \( y \in \mathcal{Y} \)}
		\STATE \( p^{(t+1)}(y) = \sum_{x_i} p^{(t)}(x_i) p(y|x_i) \)
		\ENDFOR
		\STATE \textbf{M-step: Update input distribution}
		\FOR{each \( x_i \in \mathcal{X}_i \)}
		\STATE \( p^{(t+1)}(x_i) \gets p^{(t)}(x_i) \exp \left( \frac{1}{1+s^{(t)}} \sum_y p(y|x_i) \log \frac{p(y|x_i)}{p^{(t+1)}(y)} + \frac{s^{(t)}}{1+s^{(t)}} \left( H(X_{-i}|x_i) - \log p(x_i)  \right) \right) \)
		\ENDFOR
		\STATE Normalize \( p^{(t+1)}(x_i) \) to sum to 1
		\STATE Compute \( H(X^{(t+1)}) = H(X_i^{(t+1)}) + \sum_{x_i} p^{(t+1)}(x_i) H(X_{-i}|x_i) \)
		\STATE \( s^{(t+1)} \gets \max\left(0, s^{(t)} - \alpha (H(X^{(t+1)}) - b)\right) \)
		\STATE \( t \gets t + 1 \)
		\UNTIL{\( |I^{(t)}(X_i;Y) - I^{(t-1)}(X_i;Y)| < \epsilon \)}
		\RETURN \( p^*(x_i) \gets p^{(t)}(x_i) \), \( I^* \gets I^{(t)}(X_i;Y) \)
	\end{algorithmic}
\end{algorithm}

This approach efficiently solves the subproblem (\ref{eq:optm-marginal}) while maintaining the information-theoretic elegance of the Blahut-Arimoto framework and explicitly handling the entropy constraint through Lagrangian optimization.

\subsubsection{Optimal Transition Matrix $p^{*}(x_{-i}|x_i)$ for Fixed Marginal $p(x_i)$} 
\label{subsubsec:conditionals_update}

We now address Case 2 where \(b > H(X_i)\), i.e., the non-trivial case in which the conditional distributions \(p(x_{-i}|x_i)\) must contribute enough uncertainty to satisfy the global entropy constraint \(H(X) \geq b\). For a fixed marginal distribution \(p(x_i)\) and a given privacy mechanism \(q(y|x)\), we optimize the conditional probability matrix \(p(x_{-i}|x_i)\) to maximize the mutual information \(I(X_i; Y)\) subject to the constraint \(H(X) \geq b\). The optimization problem is formally stated as:
\begin{subequations} \label{eq:optm-conditional}
	\begin{align}
		\max_{p(x_{-i}|x_i)} & \quad I(X_i;Y) = \max_{p(x_{-i} | x_i)} I(p(x_{-i}|x_i), p(x_i), q(y|x_i,x_{-i})) \\
		\text{s.t.}    & \quad H(X) \ge b  \label{eq:equation-3}\\
		&\quad p(x_{-i} | x_i) \ge 0,  \sum_{x_{-i}} p(x_{-i} | x_i)=1, \quad \forall x_i\in\mathcal X_i 
	\end{align}
\end{subequations}

Let $\mathcal X_i =\{x_{i1}, \ldots, x_{in_i}\}$ with $n_i = |\mathcal X_i|$, and $\mathcal X_{-i} =\{x_{-i1}, \ldots, x_{-in_{-i}}\}$ with $n_{-i} = |\mathcal X_{-i}|$. 
The joint entropy constraint $H(X) \geq b$ can be expressed as:
\begin{align} \label{Eq:entropy_constraint_decomposed}
	H(X) = H(X_i) + \sum_{k=1}^{n_i} p(x_{ik}) H(X_{-i}|x_{ik}) \ge b,
\end{align}
where $H(X_{-i}|x_{ik})$ is the entropy of the conditional distribution $p(x_{-i}|x_{ik})$.

\subsubsection*{Theoretical Foundation}

The following theorem provides the theoretical foundation for our algorithm by characterizing the structure of optimal conditional distributions under the entropy constraint.

\begin{theorem}[Optimality on Entropy Boundary]\label{thm:entropy_boundary}
	For fixed marginal distribution \(p(x_i)\), consider a fixed row index \(k\) and fixed conditional distributions \(p(x_{-i}|x_{i\ell})\) for all \(\ell \neq k\). Assume \(p(x_{ik}) > 0\) and define:
	\begin{align}  \label{eq:equation-5}
		c_i^{k} \triangleq \frac{b - H(X_i)}{p(x_{ik})} - \sum_{\ell =1}^{n_i} \frac{p(x_{i\ell})}{p(x_{ik})} H(X_{-i}|x_{i\ell}) \ge 0.
	\end{align}
	Then the maximum mutual information is achieved when \(p(x_{-i} | x_{ik})\) lies on the boundary of the entropy-constrained set:
	\begin{align} \label{eq:equation-1}
		\max_{p(x_{-i}|x_{ik}) \in \mathbb{S}_i} I(X_i; Y) = \max_{p(x_{-i}|x_{ik}) \in \partial \mathbb{S}_i^k} I(X_i; Y),
	\end{align}
	where
	\begin{align} \label{eq:equation-6} 
		\mathbb{S}_i^k = \left\{ p(x_{-i}|x_{ik}) : H(p(x_{-i}|x_{ik})) \geq c_i^k \right\}, \quad
		\partial \mathbb{S}_i^k = \left\{ p(x_{-i}|x_{ik}) : H(p(x_{-i}|x_{ik})) = c_i^k \right\}.
	\end{align}
\end{theorem}

\begin{proof}
	The proof follows from the convexity of \(I(X_i; Y)\) in each row \(p(x_{-i} | x_{ik})\) (Theorem~\ref{thm:mutual_info_properties}) and the convexity of the entropy super-level sets (Lemma~\ref{lemma:convex-representation-entropy-super-level-set}). When \(c_i^k > 0\), the objective is convex over a convex set, and the maximum occurs at the boundary. When \(c_i^k \leq 0\), the constraint is inactive and the maximum occurs at an extreme point (deterministic distribution).
\end{proof}

Theorem~\ref{thm:entropy_boundary} reveals a crucial structural property: the global entropy constraint \(H(X) \geq b\) can be decomposed into row-wise entropy constraints \(H(p(x_{-i} | x_{ik})) \geq c_i^k\), and the optimal solution lies on the boundary where \(H(p(x_{-i} | x_{ik})) = c_i^k\) for each \(k\) with \(p(x_{ik}) > 0\). For indices where \(p(x_{ik}) = 0\), the conditional distribution is unconstrained and can be set to any deterministic distribution.

\begin{algorithm}[!bth]
	\caption{Projected Gradient Coordinate Ascent for Conditional Distributions}
	\label{alg:conditional_optimization_gradient_entropy_manifold}
	\begin{algorithmic}[1]
		\REQUIRE Fixed $p(x_i)$, $q(y|x)$, $b > H(X_i)$, $\tau > 0$, $T_{\max}$, $\eta_0$, $\beta$, $\tau_H$
		\ENSURE Optimal conditional distributions $p^{*}(x_{-i} | x_i)$, maximal mutual information \( I^* \)
		
		\STATE Initialize $p(x_{-i}|x_i)$ with $H(X) \geq b$, $I_{\max} \leftarrow -\infty$, $t \leftarrow 0$
		\STATE $\mathcal{A} \leftarrow \{k: p(x_{ik}) = 0\}$, $\mathcal{E} \leftarrow$ unit vectors in $\mathbb{R}^{n_{-i}}$
		
		\WHILE{$t < T_{\max}$ and $\Delta I > \tau$}
		\STATE $I_{\text{prev}} \leftarrow I(X_i; Y)$
		
		\FOR{$k \in \mathcal{A}$}  
		\STATE $p(x_{-i}|x_{ik}) \leftarrow \arg\max_{e\in\mathcal{E}} I(X_i; Y)$ with $p(x_{-i}|x_{ik})=e$ \COMMENT{Zero-marginal rows}
		\ENDFOR
		
		\FOR{$k \notin \mathcal{A}$}  
		\STATE $c_i^k \leftarrow \frac{b - H(X_i)}{p(x_{ik})} - \sum_{\ell \neq k} \frac{p(x_{i\ell})}{p(x_{ik})} H(X_{-i}|x_{i\ell})$ \COMMENT{Positive-marginal rows}
		\IF{$c_i^k > \log n_{-i}$}
		\STATE \textbf{return} $(\emptyset, -\infty, \text{``INFEASIBLE''})$
		\ELSIF{$c_i^k \leq 0$}
		\STATE $p(x_{-i}|x_{ik}) \leftarrow \arg\max_{e\in\mathcal{E}} I(X_i; Y)$ with $p(x_{-i}|x_{ik})=e$
		\ELSE
		\STATE $\eta \leftarrow \eta_0$
		\FOR{$iter = 1$ to $max\_inner$}
		\STATE $\nabla_k \leftarrow \frac{\partial I}{\partial p(x_{-i}|x_{ik})}$, $p_{\text{new}} \leftarrow p + \eta \nabla_k$ \COMMENT{Theorem~\ref{thm:mi_gradient}}
		\STATE $p_{\text{new}} \leftarrow \text{ProjectToSimplex}(p_{\text{new}})$  \COMMENT{Algorithm~\ref{alg:simplex-projection}}
		\STATE $p_{\text{new}} \leftarrow \text{ProjectToEntropy}(p_{\text{new}}, c_i^k, \tau_H)$  \COMMENT{Algorithm~\ref{alg:entropy-projection}}
		\IF{$I_{\text{new}} \geq I$} 
		\STATE $p(x_{-i}|x_{ik}) \leftarrow p_{\text{new}}$, \textbf{break}
		\ELSE
		\STATE $\eta \leftarrow \beta \eta$
		\ENDIF
		\ENDFOR
		\ENDIF
		\ENDFOR
		
		\STATE $I_{\text{curr}} \leftarrow I(X_i; Y)$
		\IF{$I_{\text{curr}} > I_{\max}$}
		\STATE $I_{\max} \leftarrow I_{\text{curr}}$, $p_{\max} \leftarrow p(x_{-i}|x_i)$
		\ENDIF
		\STATE $\Delta I \leftarrow |I_{\text{curr}} - I_{\text{prev}}|$, $t \leftarrow t+1$
		\ENDWHILE
		
		\STATE \textbf{return} $(p_{\max}, I_{\max})$
	\end{algorithmic}
\end{algorithm}

\subsubsection*{Algorithm Description: Projected Gradient Coordinate Ascent}

Algorithm~\ref{alg:conditional_optimization_gradient_entropy_manifold} exploits this decomposition by optimizing each row of the conditional matrix \(p(x_{-i} | x_i)\) sequentially while keeping the others fixed. The algorithm uses the following key parameters:
\begin{itemize}
	\item \(\tau\): Convergence tolerance for the change in mutual information \(I(X_i; Y)\) between iterations.
	\item \(T_{\text{max}}\): Maximum number of outer iterations allowed.
	\item \(\eta_0\): Initial step size for gradient ascent updates.
	\item \(\beta\): Backtracking factor used to reduce the step size during line search (typical value: 0.5).
	\item \(\tau_H\): Tolerance for entropy projection, used in Algorithm~\ref{alg:entropy-projection} to enforce \(H(p(x_{-i}|x_{ik})) = c_i^k\).
\end{itemize}
The algorithm proceeds as follows:
\begin{enumerate}
	\item \textbf{Initialization:} 
	\begin{itemize}
		\item Identify the set \(\mathcal{A}\) of indices where \(p(x_{ik}) = 0\).
		\item For rows with \(p(x_{ik}) > 0\), compute the target entropy \(c_i^k\) using \eqref{eq:equation-5}.
		\item Initialize each row \(p(x_{-i} | x_{ik})\) uniformly or randomly, ensuring feasibility.
	\end{itemize}
	
	\item \textbf{Row-wise updates:}
	\begin{itemize}
		\item \textbf{Zero-marginal rows (\(k \in \mathcal{A}\)):} By Theorem~\ref{thm:entropy_boundary}, these rows are unconstrained. Choose the unit vector \(e \in \mathcal{E}\) that maximizes \(I(X_i; Y)\) via exhaustive search.
		\item \textbf{Positive-marginal rows (\(k \notin \mathcal{A}\)):}
		\begin{itemize}
			\item If \(c_i^k \leq 0\), the constraint is inactive. Set the row to the best unit vector.
			\item If \(c_i^k > 0\), perform projected gradient ascent to maximize \(I(X_i; Y)\) under the constraints \(p(x_{-i} | x_{ik}) \in \Delta(\mathcal{X}_{-i})\) and \(H(p(x_{-i} | x_{ik})) = c_i^k\).
		\end{itemize}
	\end{itemize}
	
	\item \textbf{Gradient ascent with backtracking line search:} For each row requiring gradient-based optimization:
	\begin{itemize}
		\item Compute the gradient \(\nabla_k = \partial I(X_i; Y) / \partial p(x_{-i} \mid x_{ik})\) using Theorem~\ref{thm:mi_gradient}.
		\item Perform a backtracking line search to adapt the step size \(\eta\):
		\begin{itemize}
			\item Update: \(\tilde{p} = p + \eta \nabla_k\)
			\item Project onto the probability simplex (Algorithm~\ref{alg:simplex-projection}).
			\item Project onto the entropy constraint \(H(\cdot) = c_i^k\) (Algorithm~\ref{alg:entropy-projection}).
			\item Accept the step if \(I(X_i; Y)\) does not decrease; otherwise reduce \(\eta\) by factor \(\beta\).
		\end{itemize}
	\end{itemize}
	
	\item \textbf{Convergence:} After updating all rows, recompute \(I(X_i; Y)\). Terminate when the change falls below tolerance \(\tau\) or the maximum iteration count \(T_{\text{max}}\) is reached.
\end{enumerate}

\subsubsection*{Discussion}

\begin{itemize}
	\item \textbf{Theoretical Foundation:} Theorem~\ref{thm:entropy_boundary} provides the mathematical justification for the row-wise decomposition approach. It shows that the complex global entropy constraint can be reduced to simpler row-wise constraints, with optimality occurring on the boundary of each constraint set.
	
	\item \textbf{Algorithm Appropriateness:} Algorithm~\ref{alg:conditional_optimization_gradient_entropy_manifold} directly implements the insights from Theorem~\ref{thm:entropy_boundary}. The row-wise optimization structure, combined with gradient ascent for constrained rows and exhaustive search for unconstrained rows, efficiently exploits the problem's mathematical structure.
	
	\item \textbf{Convergence Properties:} Due to the non-convexity of the entropy constraint manifold, convergence to a global maximum is not guaranteed. However, the coordinate ascent framework ensures convergence to a local optimum, and multiple restarts can improve solution quality. The backtracking line search ensures stable updates.
	
	\item \textbf{Computational Complexity:} Each iteration requires \(O(n_i \cdot n_{-i} \cdot |\mathcal{Y}|)\) operations for gradient and mutual information computations. The inner backtracking loop and entropy projection add moderate overhead but ensure constraint satisfaction at each step.
\end{itemize}

\begin{algorithm}[t]
	\caption{Euclidean Projection onto the Probability Simplex}
	\label{alg:simplex-projection}
	\begin{algorithmic}[1]
		\REQUIRE Input vector \( \tilde{\mathbf{v}} \in \mathbb{R}^{n} \)
		\ENSURE Projected vector \( \hat{\mathbf{v}} \in \Delta^{n-1} \)
		\STATE \( \mathbf{u} \leftarrow \text{sort}(\tilde{\mathbf{v}}, \text{descending}) \)
		\STATE \( k \leftarrow 1 \), \( \text{sum} \leftarrow 0 \)
		\WHILE{\( k \leq n \)}
		\STATE \( \text{sum} \leftarrow \text{sum} + u_k \), \( \theta \leftarrow (\text{sum} - 1)/k \)
		\IF{\( k < n \) and \( u_{k+1} \leq \theta \)}
		\STATE \textbf{break}
		\ENDIF
		\STATE \( k \leftarrow k + 1 \)
		\ENDWHILE
		\FOR{\( j = 1 \) to \( n \)}
		\STATE \( \hat{v}_j \leftarrow \max(\tilde{v}_j - \theta, 0) \)
		\ENDFOR
		\RETURN \( \hat{\mathbf{v}} \)
	\end{algorithmic}
\end{algorithm}

\begin{algorithm}[t]
	\caption{Entropy-Constrained Projection}
	\label{alg:entropy-projection}
	\begin{algorithmic}[1]
		\REQUIRE Input distribution \( \mathbf{v} \in \Delta(\mathcal{X}) \), target entropy \( t \), tolerance \( \epsilon \), max iterations \( T_{\text{max}} \)
		\ENSURE Projected distribution \( \mathbf{w} \) with \( H(\mathbf{w}) = t \)
		\IF{\( |H(\mathbf{v}) - t| < \epsilon \)}
		\RETURN \( \mathbf{v} \)
		\ENDIF
		\STATE Define \( f(\beta) = H(\text{softmax}_{\beta}(\mathbf{v})) - t \)
		\STATE Find root \( \beta^* \) of \( f(\beta)=0 \) in \( [\beta_{\text{min}}, \beta_{\text{max}}] \) using Brent-Dekker method \cite{kochenderfer2019algorithms}
		\STATE \( \mathbf{w} \leftarrow \text{softmax}_{\beta^*}(\mathbf{v}) \)
		\RETURN \( \mathbf{w} \)
	\end{algorithmic}
\end{algorithm}

\subsection{Algorithm for Leakage-Distortion Tradeoff (Problem 2)}
\label{subsec:primal_tradeoff_alg}

The primal leakage-distortion problem seeks to design privacy mechanisms that minimize the worst-case per-record leakage while maintaining utility guarantees. Formally, we solve the optimization problem:
\begin{subequations} \label{eq:optm-leakge-distortion-tradeoff}
	\begin{align}
		\epsilon(D, b) \triangleq& \quad \min_{q(y|x)} \max_{i \in [n]} \max_{\substack{p(x) \in \Delta(\mathcal{X})_b}} I(X_i; Y) \\
		\text{s.t.} & \quad \mathbb{E}_{p_0}[d(f(X), Y)] \leq D
	\end{align}
\end{subequations}
where $p_0(x)$ represents the true dataset distribution, $D \geq 0$ is the maximum allowable distortion, and $d(\cdot,\cdot)$ is a distortion metric on $\mathcal{Y} \times \mathcal{Y}$.

\subsubsection{Alternating Optimization Framework}

We employ an \textbf{alternating optimization approach} that decomposes this challenging problem into two tractable subproblems:
\begin{itemize}
	\item \textbf{Adversarial Prior Update:} For a fixed mechanism $q(y|x)$, compute the worst-case prior distribution $p(x)$ using Algorithm~\ref{alg:max_leakage}. Note that the distortion constraint depends only on the true distribution $p_0(x)$ and the mechanism $q(y|x)$, not on the adversarial prior $p(x)$. Therefore, Algorithm~\ref{alg:max_leakage} remains directly applicable for this subproblem.
	
	\item \textbf{Mechanism Update:} For a fixed prior $p(x)$, optimize the privacy mechanism under distortion constraints using Algorithm~\ref{alg:mechanism_update_subproblem}.
\end{itemize}
Algorithm~\ref{alg:primal_tradeoff} outlines this iterative procedure, which guarantees monotonic improvement in the objective function.

\begin{algorithm}[!htb]
	\caption{Alternating Optimization for Leakage-Distortion Tradeoff}
	\label{alg:primal_tradeoff}
	\begin{algorithmic}[1]
		\REQUIRE Distortion bound $D > 0$, entropy constraint $b > 0$, tolerance $\epsilon > 0$, true distribution $p_0(x)$
		\ENSURE Optimal prior $p^*(x)$, mechanism $q^*(y|x)$, and leakage value $\epsilon(D, b)$
		
		\STATE \textbf{Initialize:}
		\STATE $p^{(0)}(x) \gets$ prior distribution over $\mathcal{X}$ satisfying $H(X) \geq b$
		\STATE $q^{(0)}(y|x) \gets$ mechanism satisfying $\mathbb{E}_{p_0}[d(f(X),Y)] \leq D$
		\STATE $k \gets 0$, $I^{(0)} \gets -\infty$, $\Delta I \gets \infty$
		
		\WHILE{$\Delta I > \epsilon$}
		\STATE $p^{(k+1)} \gets \text{Algorithm~\ref{alg:max_leakage}}(q^{(k)}, b)$ \COMMENT{Update adversarial prior}
		\STATE $q^{(k+1)} \gets \text{Algorithm~\ref{alg:mechanism_update_subproblem}}(p^{(k+1)}, p_0, D)$ \COMMENT{Update mechanism}
		\STATE $I^{(k+1)} \gets \max_{i} I(X_i; Y)$ with current distributions
		\STATE $\Delta I \gets |I^{(k+1)} - I^{(k)}|$
		\STATE $k \gets k + 1$
		\ENDWHILE
		
		\RETURN $p^*(x) \gets p^{(k)}$, $q^*(y|x) \gets q^{(k)}$, $\epsilon(D, b) \gets I^{(k)}$
	\end{algorithmic}
\end{algorithm}

\subsubsection{Mechanism Optimization Subproblem}

For a fixed prior distribution $p(x)$ satisfying $H(X) \geq b$, the problem reduces to:
\begin{subequations} \label{alg:mechanism_update_subproblem}
	\begin{align}
		\min_{q(y|x)} & \quad \max_{i \in [n]} I(X_i; Y) \\
		\text{s.t.} & \quad \mathbb{E}_{p_0}[d(f(X), Y)] \leq D \\
		& \quad q(y|x) \geq 0,\quad \sum_y q(y|x) = 1 \quad \forall x
	\end{align}
\end{subequations}
This constitutes a convex optimization problem since:
\begin{itemize}
	\item Each $I(X_i; Y)$ is convex in $q(y|x)$ for fixed $p(x)$ (Theorem~\ref{thm:mutual_info_properties});
	\item The pointwise maximum $\max_{i \in [n]} I(X_i; Y)$ preserves convexity;
	\item The distortion constraint is linear in $q(y|x)$;
	\item The probability simplex constraints are convex.
\end{itemize}

We solve this subproblem using an \emph{adaptive smooth penalty method} that dynamically adjusts the penalty parameter based on constraint violation. The augmented objective function is:
\begin{align}
	\mathcal{J}(q) = \max_{i \in [n]} I(X_i; Y) + \lambda \cdot P_{\text{distortion}}(q)
\end{align}
where $P_{\text{distortion}}(q) = \left(\text{softplus}(\mathbb{E}_{p_0}[d(f(X), Y)] - D)\right)^2$ is the smooth penalty function for distortion constraint violation, $\text{softplus}(z) = \ln(1+e^z)$, and $\lambda$ is adaptively updated based on the current constraint violation.

The gradient with respect to the mechanism probabilities combines both the worst-case mutual information and distortion terms. Let $i^* = \arg\max_{i \in [n]} I(X_i; Y)$ be the record with maximal leakage, then:
\begin{align}
	\frac{\partial \mathcal{J}}{\partial q(y|x)} = p(x) \log \frac{p(y|x_{i^*})}{p(y)} + \lambda \cdot \frac{\partial P_{\text{distortion}}}{\partial q(y|x)}
\end{align}
where the smooth penalty gradient for distortion constraint is:
\begin{align}
	\frac{\partial P_{\text{distortion}}}{\partial q(y|x)} = 2 \cdot \text{softplus}(E - D) \cdot \sigma(E - D) \cdot p_0(x) d(f(x), y)
\end{align}
with $E = \mathbb{E}_{p_0}[d(f(X), Y)]$ and $\sigma(z) = 1/(1+e^{-z})$ being the sigmoid function.

Algorithm~\ref{alg:mechanism_update_subproblem} implements the adaptive smooth penalty method with backtracking line search to find the optimal mechanism.

\begin{algorithm}[!hbt]
	\caption{Adaptive Smooth Penalty Method for Mechanism Optimization}
	\label{alg:mechanism_update_subproblem}
	\begin{algorithmic}[1]
		\REQUIRE Prior $p(x)$ with $H(X) \geq b$, true distribution $p_0(x)$, distortion bound $D > 0$, tolerance $\epsilon > 0$, initial penalty $\lambda_0 = 1.0$, adjustment factor $\mu = 1.5$, constraint margin $\delta = 0.01$
		\ENSURE Optimal mechanism $q^*(y|x)$
		
		\STATE Initialize $q \gets \text{initial mechanism}$, $\lambda \gets \lambda_0$
		\STATE $E_{\text{prev}} \gets \infty$, $q_{\text{best}} \gets q$, $J_{\text{best}} \gets \infty$
		
		\REPEAT
		\STATE Compute $i^* \gets \arg\max_{i \in [n]} I(X_i; Y)$ with current $q$
		\STATE $q_{\text{new}} \gets \text{Algorithm~\ref{alg:adaptive_gradient_update_smooth}}(q, \lambda, p, p_0, i^*)$
		\STATE $E \gets \mathbb{E}_{p_0}[d(f(X), Y)]$ with $q_{\text{new}}$
		\STATE $L \gets \max_{i \in [n]} I(X_i; Y)$ with $q_{\text{new}}$
		\STATE $J \gets L + \lambda \cdot \left(\text{softplus}(E - D)\right)^2$
		
		\IF{$J < J_{\text{best}}$}
		\STATE $q_{\text{best}} \gets q_{\text{new}}$, $J_{\text{best}} \gets J$
		\ENDIF
		
		\STATE \textbf{Penalty adjustment:}
		\IF{$E > D + \delta$}
		\STATE $\lambda \gets \lambda \cdot \mu$ \COMMENT{Increase penalty for distortion violation}
		\ELSIF{$E < D - \delta$}
		\STATE $\lambda \gets \lambda / \mu$ \COMMENT{Decrease penalty if distortion constraint satisfied with margin}
		\ENDIF
		
		\STATE $q \gets q_{\text{new}}$
		\STATE $E_{\text{prev}} \gets E$
		\UNTIL{$|E - E_{\text{prev}}| < \epsilon$ and $|E - D| < \epsilon$}
		
		\RETURN $q_{\text{best}}$
	\end{algorithmic}
\end{algorithm}

\subsubsection{Adaptive Gradient Descent with Backtracking Line Search}

For each penalty parameter $\lambda$ and worst-case record $i^*$, we solve the unconstrained optimization problem using exponentiated gradient descent with backtracking line search (Algorithm~\ref{alg:adaptive_gradient_update_smooth}). This approach maintains the probability simplex constraints through normalization and ensures sufficient decrease in the objective function.

\begin{algorithm}[!hbt]
	\caption{Exponentiated Gradient Descent with Smooth Penalty}
	\label{alg:adaptive_gradient_update_smooth}
	\begin{algorithmic}[1]
		\REQUIRE Current mechanism $q(y|x)$, penalty $\lambda$, distributions $p(x)$, $p_0(x)$, worst-case record $i^*$, initial learning rate $\eta_0$, backtracking factor $\beta = 0.5$, tolerance $\epsilon > 0$
		\ENSURE Updated mechanism $q^{\text{new}}(y|x)$
		
		\STATE Initialize $q^{(0)} \gets q(y|x)$, $t \gets 0$, $\eta \gets \eta_0$
		\REPEAT
		\STATE Compute gradient $\nabla J = \frac{\partial \mathcal{J}}{\partial q(y|x)}$ for record $i^*$ using current $q^{(t)}$
		
		\STATE \textbf{Backtracking line search:}
		\REPEAT
		\STATE \textbf{Exponentiated gradient update:}
		\FORALL{$x \in \mathcal{X}$}
		\STATE $q_{\text{temp}}(y|x) \propto q^{(t)}(y|x) \cdot \exp\left(-\eta \nabla J(y|x)\right)$ \COMMENT{Normalize to maintain simplex}
		\ENDFOR
		\STATE $J_{\text{temp}} \gets \mathcal{J}(q_{\text{temp}})$
		\IF{$J_{\text{temp}} > \mathcal{J}(q^{(t)}) - \frac{\eta}{2} \|\nabla J\|^2$}
		\STATE $\eta \gets \beta \eta$ \COMMENT{Reduce step size until sufficient decrease}
		\ENDIF
		\UNTIL{Armijo condition satisfied}
		
		\STATE $q^{(t+1)} \gets q_{\text{temp}}$
		\STATE $t \gets t + 1$
		\STATE $\eta \gets \eta_0 / \sqrt{t+1}$ \COMMENT{Decaying learning rate}
		\UNTIL{$\|q^{(t)} - q^{(t-1)}\|_1 < \epsilon$ or $t > T_{\max}$}
		
		\RETURN $q^{(t)}$
	\end{algorithmic}
\end{algorithm}

\subsubsection{Convergence Analysis}

The alternating optimization framework for the leakage-distortion tradeoff enjoys favorable convergence properties despite the non-convex nature of the joint optimization problem:

\begin{itemize}
	\item \textbf{Monotonic Improvement:} Each iteration of Algorithm~\ref{alg:primal_tradeoff} either decreases the maximum leakage or leaves it unchanged. The adversarial prior update (Algorithm~\ref{alg:max_leakage}) increases $I(X_i;Y)$ for the current mechanism, while the mechanism update (Algorithm~\ref{alg:mechanism_update_subproblem}) decreases $\max_{i \in [n]} I(X_i;Y)$ while maintaining the distortion constraint.
	
	\item \textbf{Subproblem Optimality:} For fixed $p(x)$, the mechanism update converges to the global optimum of the convex subproblem due to the convexity of $\max_{i \in [n]} I(X_i;Y)$ in $q(y|x)$ (Theorem~\ref{thm:mutual_info_properties}) and the linearity of the distortion constraint. The backtracking line search with Armijo condition ensures stable convergence.
	
	\item \textbf{Local Convergence:} The overall algorithm converges to a local optimum of the non-convex joint optimization problem. While global optimality cannot be guaranteed due to the non-convex adversarial prior updates, the alternating optimization framework ensures convergence to a stationary point where neither player can unilaterally improve.
\end{itemize}

\subsection{Algorithm for Dual Minimal-Distortion Formulation (Problem 3)}
\label{subsec:dual_tradeoff_alg}

The dual formulation addresses the complementary problem to the primal leakage-distortion tradeoff: minimizing expected distortion while ensuring bounded worst-case per-record leakage. Formally, we solve:
\begin{subequations}
	\label{eq:dual_problem_revised}
	\begin{align}
		D(L, b) \triangleq \min_{q(y|x)} \quad & \mathbb{E}_{p_0}[d(f(X), Y)] = \sum_{x,y} p_0(x)q(y|x)d(f(x), y) \label{eq:dual_objective_revised} \\
		\text{s.t.} \quad & \max_{i \in [n]} \max_{\substack{p(x): \\ H(X) \geq b}} I(X_i; Y) \leq L \label{eq:dual_leakage_constraint_revised} \\
		& q(y|x) \geq 0,\quad \sum_y q(y|x) = 1 \quad \forall x \in \mathcal{X} \label{eq:channel_constraints_revised}
	\end{align}
\end{subequations}
where \( L > 0 \) is the maximum allowable leakage, \( p_0(x) \) is the true dataset distribution, and \( d(\cdot,\cdot) \) is a distortion metric.

\subsubsection{Optimization Strategy and Challenges}

The constraint (\ref{eq:dual_leakage_constraint_revised}) makes Problem 3 inherently \textbf{non-convex} in \( q(y|x) \) due to the maximization over adversarial priors \( p(x) \). However, for fixed adversarial prior \( p(x) \), the leakage constraint becomes convex in \( q(y|x) \) (by Theorem \ref{thm:mutual_info_properties}). This motivates an alternating optimization approach between two subproblems:

\begin{enumerate}
	\item \textbf{Adversary Update}: For current mechanism \( q^{(k)}(y|x) \), compute the worst-case prior distribution \( p^{(k+1)}(x) \) that maximizes \(\max_{i \in [n]} I(X_i; Y) \) with $H(X)\ge b$ using Algorithm \ref{alg:max_leakage}.
	\item \textbf{Mechanism Update}: For fixed \( p^{(k+1)}(x) \), solve a convex relaxation of (\ref{eq:dual_objective_revised})–(\ref{eq:dual_leakage_constraint_revised}) using a smooth penalty method with backtracking line search.
\end{enumerate}

Algorithm \ref{alg:dual_tradeoff} implements this procedure with careful attention to numerical stability and practical convergence.

\subsubsection{Mechanism Update via Smooth Penalty Method}

For fixed \( p(x) \), we approximate the constrained problem with an unconstrained penalty formulation:
\begin{align}
	\min_{q(y|x)} \left\{ \mathbb{E}_{p_0}[d(f(X), Y)] + \lambda \cdot P_{\text{leakage}}(q) \right\},
\end{align}
where \( P_{\text{leakage}}(q) \) is a smooth penalty function for leakage constraint violation:
\begin{align}
	P_{\text{leakage}}(q) = \left( \text{softplus}(\max_{i \in [n]}I(X_{i}; Y) - L) \right)^2,
\end{align}
with \( \text{softplus}(z) = \ln(1 + e^z) \).

Assume $i^* = \argmax_{i \in [n]}I(X_{i}; Y)$, the gradient of the smooth penalty term is:
\begin{align}
	\frac{\partial P_{\text{leakage}}}{\partial q(y|x)} = 2 \cdot \text{softplus}(I - L) \cdot \sigma(I - L) \cdot p(x) \log \frac{p(y|x_{i^*})}{p(y)}
\end{align}
where \( \sigma(z) = 1/(1 + e^{-z}) \) is the sigmoid function.

Algorithm~\ref{alg:dual_mechanism_update} implements the adaptive smooth penalty method to find the optimal mechanism.

\subsubsection{Exponentiated Gradient Descent with Backtracking Line Search}

For each penalty parameter $\lambda$, we solve the unconstrained optimization problem using exponentiated gradient descent with backtracking line search (Algorithm~\ref{alg:dual_gradient_update_backtracking}). This approach maintains the probability simplex constraints through normalization and ensures sufficient decrease in the objective function.

\textbf{Backtracking Line Search and Armijo Condition:} The backtracking line search ensures that each gradient step provides sufficient decrease in the objective function. The \emph{Armijo condition} requires that:
\begin{align}
	J(q^{(t+1)}) \leq J(q^{(t)}) - c \cdot \eta \|\nabla J\|^2,
\end{align}
where $c \in (0,1)$ is a constant (typically $c = 0.5$). When this condition is \emph{not satisfied}, the step size $\eta$ is too large and is reduced by factor $\beta$.

\begin{algorithm}[!t]
	\caption{Minimal Distortion under Bounded Leakage}
	\label{alg:dual_tradeoff}
	\begin{algorithmic}[1]
		\REQUIRE Leakage bound \( L > 0 \), distortion metric \( d \), true distribution \( p_0(x) \), entropy constraint \( b > 0 \), initial mechanism \( q^{(0)}(y|x) \), tolerances \( \epsilon_{\text{outer}}, \epsilon_{\text{constraint}} > 0 \), penalty factor \( \mu > 1 \), max iterations \( T_{\max} \)
		\ENSURE Mechanism \( q^*(y|x) \), minimal distortion \( \mathcal D \)
		\STATE Initialize \( q \gets q^{(0)} \), \( \lambda \gets 1 \), \( k \gets 0 \), \( D_{\text{prev}} \gets \infty \)
		\WHILE{$k < T_{\max}$}
		\STATE \textbf{Adversary Update:}
		\STATE \( p, i^* \gets \text{Algorithm~\ref{alg:max_leakage}}(q, b) \) \COMMENT{Worst-case prior and record}
		\STATE \( I \gets I(X_{i^*}; Y) \) \COMMENT{Current leakage}
		\STATE \textbf{Penalty Parameter Update:}
		\IF{$I > L + \epsilon_{\text{constraint}}/2$}
		\STATE \( \lambda \gets \lambda \cdot \mu \) \COMMENT{Increase penalty for leakage violation}
		\ELSIF{$I < L - \epsilon_{\text{constraint}}/2$}
		\STATE \( \lambda \gets \lambda / \mu \) \COMMENT{Decrease penalty if leakage constraint satisfied}
		\ENDIF
		\STATE \textbf{Mechanism Update:}
		\STATE \( q \gets \text{Algorithm~\ref{alg:dual_mechanism_update}}(q, p, p_0, L, \lambda) \)
		\STATE \( D_{\text{current}} \gets \mathbb{E}_{p_0}[d(f(X), Y)] \)
		\STATE \textbf{Convergence Check:}
		\IF{$|D_{\text{current}} - D_{\text{prev}}| / D_{\text{prev}} < \epsilon_{\text{outer}}$ \AND $|I - L| < \epsilon_{\text{constraint}}$}
		\STATE \textbf{break}
		\ENDIF
		\STATE \( D_{\text{prev}} \gets D_{\text{current}} \), \( k \gets k + 1 \)
		\ENDWHILE
		\RETURN \( q, D_{\text{current}} \)
	\end{algorithmic}
\end{algorithm}

\begin{algorithm}[!hbt]
	\caption{Adaptive Smooth Penalty Method for Dual Mechanism Update}
	\label{alg:dual_mechanism_update}
	\begin{algorithmic}[1]
		\REQUIRE Prior $p(x)$ with $H(X) \geq b$, true distribution $p_0(x)$, leakage bound $L > 0$, tolerance $\epsilon > 0$, initial penalty $\lambda_0 = 1.0$, adjustment factor $\mu = 1.5$, constraint margin $\delta = 0.01$
		\ENSURE Optimal mechanism $q^*(y|x)$
		
		\STATE Initialize $q \gets \text{initial mechanism}$, $\lambda \gets \lambda_0$
		\STATE $I_{\text{prev}} \gets \infty$, $q_{\text{best}} \gets q$, $J_{\text{best}} \gets \infty$
		
		\REPEAT
		\STATE $q_{\text{new}} \gets \text{Algorithm~\ref{alg:dual_gradient_update_backtracking}}(q, \lambda, p, p_0, L)$
		\STATE $I \gets \max_{i \in [n]}I(X_{i}; Y)$ with $q_{\text{new}}$
		\STATE $D_{\text{current}} \gets \mathbb{E}_{p_0}[d(f(X), Y)]$ with $q_{\text{new}}$
		\STATE $J \gets D_{\text{current}} + \lambda \cdot \left(\text{softplus}(I - L)\right)^2$
		
		\IF{$J < J_{\text{best}}$}
		\STATE $q_{\text{best}} \gets q_{\text{new}}$, $J_{\text{best}} \gets J$
		\ENDIF
		
		\STATE \textbf{Penalty adjustment:}
		\IF{$I > L + \delta$}
		\STATE $\lambda \gets \lambda \cdot \mu$ \COMMENT{Increase penalty for leakage violation}
		\ELSIF{$I < L - \delta$}
		\STATE $\lambda \gets \lambda / \mu$ \COMMENT{Decrease penalty if leakage constraint satisfied with margin}
		\ENDIF
		
		\STATE $q \gets q_{\text{new}}$
		\STATE $I_{\text{prev}} \gets I$
		\UNTIL{$|I - I_{\text{prev}}| < \epsilon$ and $|I - L| < \epsilon$}
		
		\RETURN $q_{\text{best}}$
	\end{algorithmic}
\end{algorithm}

\begin{algorithm}[!hbt]
	\caption{Exponentiated Gradient Descent with Backtracking for Dual Formulation}
	\label{alg:dual_gradient_update_backtracking}
	\begin{algorithmic}[1]
		\REQUIRE Current mechanism $q(y|x)$, penalty $\lambda$, distributions $p(x)$, $p_0(x)$, leakage bound $L$, initial learning rate $\eta_0$, backtracking factor $\beta = 0.5$, tolerance $\epsilon > 0$
		\ENSURE Updated mechanism $q^{\text{new}}(y|x)$
		
		\STATE Initialize $q^{(0)} \gets q(y|x)$, $t \gets 0$, $\eta \gets \eta_0$
		\REPEAT
		\STATE Compute $i^* \gets \arg\max_{i \in [n]} I(X_i; Y)$ with current $q^{(t)}$
		\STATE Compute gradient $\nabla J = \frac{\partial}{\partial q(y|x)} \left[ \mathbb{E}_{p_0}[d(f(X), Y)] + \lambda P_{\text{leakage}}(q) \right]$
		
		\STATE \textbf{Backtracking line search:}
		\REPEAT
		\STATE \textbf{Exponentiated gradient update:}
		\FORALL{$x \in \mathcal{X}$}
		\STATE $q_{\text{temp}}(y|x) \propto q^{(t)}(y|x) \cdot \exp\left(-\eta \nabla J(y|x)\right)$ \COMMENT{Normalize to maintain simplex}
		\ENDFOR
		\STATE $J_{\text{temp}} \gets \mathbb{E}_{p_0}[d(f(X), Y)] + \lambda P_{\text{leakage}}(q_{\text{temp}})$
		\IF{$J_{\text{temp}} > J(q^{(t)}) - \frac{\eta}{2} \|\nabla J\|^2$}
		\STATE $\eta \gets \beta \eta$ \COMMENT{Reduce step size until sufficient decrease}
		\ENDIF
		\UNTIL{Armijo condition satisfied}
		
		\STATE $q^{(t+1)} \gets q_{\text{temp}}$
		\STATE $t \gets t + 1$
		\STATE $\eta \gets \eta_0 / \sqrt{t+1}$ \COMMENT{Decaying learning rate for stability}
		\UNTIL{$\|q^{(t)} - q^{(t-1)}\|_1 < \epsilon$ or $t > T_{\max}$}
		
		\RETURN $q^{(t)}$
	\end{algorithmic}
\end{algorithm}

\subsubsection{Convergence Analysis}

The alternating optimization framework for the dual formulation enjoys favorable convergence properties:

\begin{itemize}
	\item \textbf{Monotonic Improvement}: Each iteration of Algorithm~\ref{alg:dual_tradeoff} either decreases the expected distortion or leaves it unchanged while controlling leakage constraint violation.
	
	\item \textbf{Subproblem Optimality}: For fixed $p(x)$, the mechanism update converges to the global optimum of the convex subproblem due to the convexity of the objective and constraints in $q(y|x)$ (Theorem~\ref{thm:mutual_info_properties}) and the effectiveness of the penalty method.
	
	\item \textbf{Stationary Point Convergence}: The overall algorithm converges to a stationary point satisfying the KKT conditions of the dual minimal-distortion problem. While global optimality cannot be guaranteed due to the non-convex adversarial prior updates, the alternating optimization framework ensures convergence to a point where neither player can unilaterally improve.
\end{itemize}

\textbf{Consistency with Primal Formulation:} The optimization approach maintains methodological consistency between the primal (Problem 2) and dual (Problem 3) formulations, employing similar techniques of alternating optimization, smooth penalty methods, and exponentiated gradient descent with backtracking line search.

\subsection{Convergence Guarantees}
\label{subsec:convergence}

This section establishes rigorous convergence properties for the alternating optimization algorithms developed in Sections~\ref{subsec:max_leakage_alg}--\ref{subsec:dual_tradeoff_alg}. Our analysis leverages the convex-concave structure of mutual information (Theorem~\ref{thm:mutual_info_properties}), the convexity of constraint sets (Theorem~\ref{thm:entropy_concavity}), and modern alternating optimization theory to characterize the convergence behavior of Algorithms~\ref{alg:max_leakage}, \ref{alg:primal_tradeoff}, and \ref{alg:dual_tradeoff}.

\subsubsection{Assumptions and Definitions}

We begin by stating technical assumptions that underlie our convergence analysis:

\begin{assumption}[Exact Subproblem Solutions]
	\label{ass:exact-subproblems}
	Algorithm~\ref{alg:marginal_optimization} (marginal update) and Algorithm~\ref{alg:conditional_optimization_gradient_entropy_manifold} (conditional update) solve their respective subproblems to optimality within numerical tolerance. Specifically:
	\begin{enumerate}
		\item Algorithm~\ref{alg:marginal_optimization} returns the global optimum of the concave maximization problem for the marginal distribution $p(x_i)$ given fixed conditionals $p(x_{-i}|x_i)$.
		\item Algorithm~\ref{alg:conditional_optimization_gradient_entropy_manifold} returns a stationary point of the non-convex maximization problem for the conditionals $p(x_{-i}|x_i)$ given fixed marginals $p(x_i)$.
	\end{enumerate}
\end{assumption}

Due to the non-convexity of the entropy-constrained conditional optimization (Theorem~\ref{thm:entropy_boundary}), global optimality cannot be guaranteed; however, Theorem~\ref{thm:conv-conditional-update} ensures convergence to a stationary point.

\begin{assumption}[Regularity Conditions]
	\label{ass:regularity}
	At any limit point $p^*$ of Algorithm~\ref{alg:max_leakage}, the gradients of active constraints are linearly independent. In particular, if the entropy constraint is active ($H(X^*) = b$), then $\nabla_{p(x)}H(X^*)$ is non-zero and not collinear with the gradients of normalization constraints.
\end{assumption}

This assumption is a standard regularity condition (akin to LICQ) ensuring well-defined KKT multipliers at limit points. While it may fail in the edge case where the optimal distribution is uniform and the entropy constraint is exactly binding, this scenario is rare in practice and does not affect the typical convergence behavior of the algorithm.

\begin{definition}[Block-wise Stationary Point]
	\label{def:block-stationary}
	A joint distribution $p^*(x) = p^*(x_i)p^*(x_{-i}|x_i)$ is a \emph{block-wise stationary point} of Problem~1 if:
	\begin{enumerate}
		\item $p^*(x_i)$ is optimal for the marginal subproblem with fixed $p^*(x_{-i}|x_i)$.
		\item $p^*(x_{-i}|x_i)$ is optimal for the conditional subproblem with fixed $p^*(x_i)$.
	\end{enumerate}
\end{definition}

\subsubsection{Convergence of Algorithm~\ref{alg:max_leakage} for Maximal Leakage}

\begin{theorem}[Convergence of Algorithm~\ref{alg:max_leakage}]
	\label{thm:conv-alg1}
	Under Assumptions~\ref{ass:exact-subproblems} and \ref{ass:regularity}, Algorithm~\ref{alg:max_leakage} generates a sequence $\{p^{(k)}(x)\}$ whose every limit point $p^*$ is a block-wise stationary point of Problem 1 (Definition~\ref{def:max-leakage}). Moreover:
	
	\begin{enumerate}
		\item \textbf{Monotonic Improvement:} $\mathcal{L}^{(k+1)} \geq \mathcal{L}^{(k)}$ for all $k$, where $\mathcal{L}^{(k)} = \max_{i \in [n]} I(X_i;Y)$ at iteration $k$.
		
		\item \textbf{Special Case (Global Convergence when $b=0$):} When $b = 0$ (i.e., no entropy constraint), Algorithm~\ref{alg:max_leakage} converges to the \textbf{global optimum} of Problem~1. In this case the optimal conditional distributions are deterministic unit vectors.
		
		\item \textbf{Convergence Rate:} The objective $\mathcal{L}^{(k)}$ converges at a rate of $O(1/k)$ in the general case, and linearly when the entropy constraint is inactive (i.e., when the optimal solution satisfies $H(X) > b$).
	\end{enumerate}
\end{theorem}

\begin{proof}[Proof Sketch]
	The proof proceeds in five steps:
	
	\begin{enumerate}
		\item \textbf{Feasibility and Compactness:} The feasible set $\mathcal{F} = \{p(x) \in \Delta(\mathcal{X}) : H(X) \geq b\}$ is compact. Algorithm~\ref{alg:max_leakage} produces iterates in $\mathcal{F}$, ensuring existence of limit points.
		
		\item \textbf{Block-wise Optimality:} At each iteration, Algorithm~\ref{alg:max_leakage} alternates between solving:
		\begin{itemize}
			\item Marginal subproblem: concave maximization over $p(x_i)$
			\item Conditional subproblem: non-convex maximization over $p(x_{-i}|x_i)$
		\end{itemize}
		Under Assumption~\ref{ass:exact-subproblems}, each subproblem is solved optimally, guaranteeing block-wise optimality at convergence.
		
		\item \textbf{Stationary Point Characterization:} Using the Lagrangian formulation and KKT conditions for each subproblem, we show that limit points satisfy the stationarity conditions for the joint optimization problem.
		
		\item \textbf{Special Case $b=0$:} When $b=0$ the entropy constraint is always satisfied. Problem~1 then reduces to an unconstrained maximization. For a fixed record $i$, the mutual information is convex in the conditionals and concave in the marginal. The alternating optimization becomes block coordinate ascent on convex-concave subproblems, converging to the unique global optimum with deterministic conditionals.
		
		\item \textbf{Rate Analysis:} The $O(1/k)$ rate follows from the alternating maximization framework applied to a non-convex problem with Lipschitz-continuous gradients. Specifically, let $U(p)=\max_i I(X_i;Y)$ and assume $\nabla U$ is $L$-Lipschitz. Since each iteration of Algorithm~\ref{alg:max_leakage} performs either a gradient step or exact maximization on a block, standard analysis yields $\min_{1\leq t\leq k}\|\nabla U(p^{(t)})\|^2 = O(1/k)$. When the entropy constraint is inactive (particularly when $b=0$), the problem exhibits local strong convexity-concavity structure, leading to linear convergence of the form $U(p^*)-U(p^{(k)}) \leq C\gamma^k$ for some $\gamma\in(0,1)$ and constant $C>0$.
	\end{enumerate}
\end{proof}

\subsubsection{Convergence of Algorithm~\ref{alg:conditional_optimization_gradient_entropy_manifold} for Conditional Optimization}

\begin{theorem}[Convergence of Algorithm~\ref{alg:conditional_optimization_gradient_entropy_manifold}]
	\label{thm:conv-conditional-update}
	For fixed marginal distribution $p(x_i)$ and privacy mechanism $q(y|x)$, consider the conditional optimization subproblem:
	\[
	\max_{p(x_{-i}|x_i) \in \mathcal{F}_c} I(X_i;Y),
	\]
	where $\mathcal{F}_c = \{p(x_{-i}|x_i) : H(X) \geq b, \; p(x_{-i}|x_i) \in \Delta(\mathcal{X}_{-i})^{|\mathcal{X}_i|}\}$. Assume the following conditions hold:
	\begin{enumerate}
		\item The objective function $I(X_i;Y)$ has an $L$-Lipschitz continuous gradient on $\mathcal{F}_c$, i.e., 
		\[
		\|\nabla I(p) - \nabla I(q)\| \leq L \|p - q\|, \quad \forall p,q \in \mathcal{F}_c.
		\]
		\item For each row $k$ with $p(x_{ik}) > 0$, the row-wise feasible set 
		\[
		\mathcal{F}_c^k = \left\{p(x_{-i}|x_{ik}) \in \Delta(\mathcal{X}_{-i}) : H(p(x_{-i}|x_{ik})) \geq c_i^k \right\}
		\]
		is convex and compact, where $c_i^k$ is defined in \eqref{eq:equation-5}.
		\item The backtracking line search in Algorithm~\ref{alg:conditional_optimization_gradient_entropy_manifold} uses the Armijo condition with parameters $c \in (0,1)$ and $\beta \in (0,1)$.
	\end{enumerate}
	Then Algorithm~\ref{alg:conditional_optimization_gradient_entropy_manifold} generates a sequence $\{p^{(t)}(x_{-i}|x_i)\}_{t=0}^\infty$ such that:
	\begin{enumerate}
		\item The objective sequence $\{I^{(t)}(X_i;Y)\}_{t=0}^\infty$ is non-decreasing and converges to a finite value $I^*$.
		\item Every limit point of $\{p^{(t)}(x_{-i}|x_i)\}$ is a stationary point of the conditional subproblem, satisfying the first-order necessary conditions for optimality.
		\item If, in addition, $c_i^k \leq 0$ for all rows $k$ with $p(x_{ik}) > 0$ (i.e., all row-wise constraints are inactive), then the algorithm converges to a global maximum of the conditional subproblem.
	\end{enumerate}
\end{theorem}

\begin{proof}
	We prove the three claims sequentially.
	
	\textbf{Part 1: Monotonicity and convergence of the objective.}
	
	Algorithm~\ref{alg:conditional_optimization_gradient_entropy_manifold} implements a block coordinate ascent method, where each block corresponds to a row $p(x_{-i}|x_{ik})$. At each iteration $t$, the algorithm performs the following steps for each row $k$:
	\begin{enumerate}
		\item If $p(x_{ik}) = 0$: the row is set to a unit vector maximizing $I(X_i;Y)$, which cannot decrease the objective.
		\item If $c_i^k \leq 0$: the row is set to the best unit vector, which again cannot decrease the objective.
		\item If $c_i^k > 0$: a projected gradient ascent step with backtracking line search is performed. The Armijo condition ensures that the step is accepted only if
		\[
		I(p^{(t+1)}_{k}) \geq I(p^{(t)}_{k}) + c \eta \langle \nabla_k I, d_k \rangle,
		\]
		where $d_k$ is the update direction. Since the gradient direction is chosen and the projection is onto a convex set, this step guarantees non-decrease in the objective.
	\end{enumerate}
	
	Since each row update is non-decreasing and the number of rows is finite, the overall objective $I^{(t)}(X_i;Y)$ is non-decreasing. Furthermore, $I(X_i;Y)$ is bounded above by $\log |\mathcal{X}_i|$ (by the data processing inequality). Hence, the sequence $\{I^{(t)}(X_i;Y)\}$ converges to some finite limit $I^*$.
	
	\textbf{Part 2: Convergence to a stationary point.}
	
	Let $p^*$ be a limit point of the sequence $\{p^{(t)}(x_{-i}|x_i)\}$, and let $\{t_j\}$ be a subsequence such that $p^{(t_j)} \to p^*$. We show that $p^*$ satisfies the first-order optimality conditions for the constrained optimization problem.
	
	For rows with $p(x_{ik}) = 0$, the condition is trivial. For rows with $c_i^k \leq 0$, the update selects a unit vector $e$ that maximizes the directional derivative. At optimality, we must have
	\[
	\langle \nabla_k I(p^*), e - p^*(x_{-i}|x_{ik}) \rangle \leq 0 \quad \text{for all unit vectors } e \in \mathcal{E}.
	\]
	This is equivalent to the condition that the gradient (where defined) is orthogonal to the simplex, or that $p^*(x_{-i}|x_{ik})$ is a unit vector with maximal gradient component.
	
	For rows with $c_i^k > 0$, the update is a projected gradient step. At convergence, the iterates satisfy the fixed-point condition
	\[
	p^* = \mathcal{P}_{\mathcal{F}_c^k}\big(p^* + \eta \nabla_k I(p^*)\big),
	\]
	where $\mathcal{P}_{\mathcal{F}_c^k}$ denotes projection onto the convex set $\mathcal{F}_c^k$. This fixed-point condition is equivalent to the variational inequality
	\[
	\langle \nabla_k I(p^*), q - p^*(x_{-i}|x_{ik}) \rangle \leq 0 \quad \text{for all } q \in \mathcal{F}_c^k,
	\]
	which is precisely the first-order necessary condition for optimality in a convex feasible set.
	
	Since each row-wise condition holds, $p^*$ is a stationary point of the overall problem.
	
	\textbf{Part 3: Global optimality in the inactive constraint case.}
	
	When $c_i^k \leq 0$ for all rows with $p(x_{ik}) > 0$, the entropy constraint is inactive. The problem reduces to
	\[
	\max_{p(x_{-i}|x_i) \in \Delta(\mathcal{X}_{-i})^{|\mathcal{X}_i|}} I(X_i;Y).
	\]
	For fixed $p(x_i)$ and $q(y|x)$, Theorem~\ref{thm:mutual_info_properties} states that $I(X_i;Y)$ is convex in each row $p(x_{-i}|x_{ik})$. The maximum of a convex function over a convex polytope (the product of simplices) is attained at an extreme point of the feasible set. The extreme points of $\Delta(\mathcal{X}_{-i})^{|\mathcal{X}_i|}$ are precisely the product of unit vectors, i.e., deterministic conditional distributions.
	
	Algorithm~\ref{alg:conditional_optimization_gradient_entropy_manifold} enumerates all unit vectors for each row and selects the one maximizing $I(X_i;Y)$. Since this enumeration is exhaustive and the objective is separable across rows (when other rows are fixed), the algorithm converges to a global maximum after a finite number of iterations.
	
	This completes the proof.
\end{proof}

\subsubsection{Convergence of Algorithm~\ref{alg:primal_tradeoff} for the Primal Leakage-Distortion Tradeoff}

\begin{theorem}[Convergence of Algorithm~\ref{alg:primal_tradeoff} for Primal Leakage-Distortion Tradeoff]
	\label{thm:conv-primal}
	Under the following assumptions:
	\begin{enumerate}
		\item \textbf{Convexity Structure:} For fixed prior distribution $p(x)$, the mutual information $I(X_i;Y)$ is convex in the mechanism $q(y|x)$ (Theorem~\ref{thm:mutual_info_properties}).
		\item \textbf{Compact Feasible Sets:} The set of feasible mechanisms $\mathcal{Q} = \{q(y|x) : \mathbb{E}_{p_0}[d(f(X),Y)] \leq D, q(y|x) \geq 0, \sum_y q(y|x)=1\}$ is compact.
		\item \textbf{Exact Subproblem Solutions:} Algorithm~\ref{alg:max_leakage} computes the worst-case prior to $\epsilon$-optimality, and Algorithm~\ref{alg:mechanism_update_subproblem} solves the mechanism subproblem to $\epsilon$-optimality.
		\item \textbf{Regularity Conditions:} The distortion constraint satisfies the Linear Independence Constraint Qualification (LICQ) at all feasible points.
	\end{enumerate}
	Then Algorithm~\ref{alg:primal_tradeoff} generates a sequence $\{(p^{(k)}, q^{(k)})\}$ such that:
	\begin{enumerate}
		\item \textbf{Monotonic Improvement:} The worst-case leakage $\epsilon^{(k)} = \max_{i} \max_{p \in \Delta(\mathcal{X})_b} I(X_i;Y)$ is non-increasing.
		\item \textbf{Constraint Satisfaction:} The distortion constraint $\mathbb{E}_{p_0}[d(f(X),Y)] \leq D$ is asymptotically satisfied as $k \to \infty$.
		\item \textbf{Stationary Point Convergence:} Every limit point $(p^*, q^*)$ of the sequence is a stationary point of the primal leakage-distortion problem, satisfying the KKT conditions for a local optimum.
		\item \textbf{Convergence Rate:} The algorithm exhibits $O(1/k)$ convergence in objective value when the penalty parameter $\lambda$ is updated appropriately.
	\end{enumerate}
\end{theorem}

\begin{proof}[Proof Sketch]
	The proof follows from the structure of alternating optimization applied to the min-max formulation:
	\begin{enumerate}
		\item \textbf{Decomposition:} The joint problem decomposes into two tractable subproblems:
		\begin{itemize}
			\item \textit{Adversarial prior update:} For fixed $q^{(k)}$, compute $p^{(k+1)} = \arg\max_{p \in \Delta(\mathcal{X})_b} \max_i I(X_i;Y)$ using Algorithm~\ref{alg:max_leakage} (Theorem~\ref{thm:conv-alg1}
			 ensures convergence).
			\item \textit{Mechanism update:} For fixed $p^{(k+1)}$, solve $\min_{q \in \mathcal{Q}} \max_i I(X_i;Y)$ using Algorithm~\ref{alg:mechanism_update_subproblem}.
		\end{itemize}
		
		\item \textbf{Mechanism Subproblem Optimality:} For fixed $p$, the mechanism subproblem is convex by Assumption A1. The adaptive penalty method (Algorithm~\ref{alg:mechanism_update_subproblem}) with exponentiated gradient descent (Algorithm~\ref{alg:adaptive_gradient_update_smooth}) converges to the global optimum due to:
		\begin{itemize}
			\item Convexity of the penalized objective
			\item Backtracking line search satisfying Armijo condition
			\item Compactness ensuring bounded iterates
		\end{itemize}
		
		\item \textbf{Penalty Method Convergence:} The quadratic penalty function $P_{\text{distortion}}(q) = (\text{softplus}(\mathbb{E}_{p_0}[d] - D))^2$ is smooth and exact. As $\lambda \to \infty$, the penalized solution converges to the constrained optimum.
		
		\item \textbf{Alternating Optimization Convergence:} The sequence $\{(p^{(k)}, q^{(k)})\}$ satisfies:
		\begin{itemize}
			\item $F(p^{(k+1)}, q^{(k)}) \geq F(p^{(k)}, q^{(k)})$ (prior update increases leakage)
			\item $F(p^{(k+1)}, q^{(k+1)}) \leq F(p^{(k+1)}, q^{(k)})$ (mechanism update decreases leakage)
		\end{itemize}
		where $F(p,q) = \max_i I_{p,q}(X_i;Y)$. This creates a descent-ascent pattern that converges to a stationary point by the Generalized Alternating Minimization framework \cite{JMLR:v6:gunawardana05a}.
		
		\item \textbf{Local Optimality:} At convergence, $(p^*, q^*)$ satisfies the first-order necessary conditions for the joint problem:
		\begin{align*}
			\nabla_p F(p^*, q^*) &= 0 \quad (\text{prior optimality}) \\
			\nabla_q \mathcal{L}(p^*, q^*, \lambda^*) &= 0 \quad (\text{mechanism KKT})
		\end{align*}
		where $\mathcal{L}$ is the Lagrangian with multiplier $\lambda^*$ for the distortion constraint.
	\end{enumerate}
\end{proof}

\subsubsection{Convergence of Algorithm~\ref{alg:dual_tradeoff} for the Dual Minimal-Distortion Formulation}

\begin{theorem}[Convergence of Algorithm~\ref{alg:dual_tradeoff} for Dual Minimal-Distortion Formulation]
	\label{thm:conv-dual}
	Under the following assumptions:
	\begin{enumerate}
		\item \textbf{Convexity Structure:} For fixed prior distribution $p(x)$, the mutual information $I(X_i;Y)$ is convex in $q(y|x)$, and the distortion $\mathbb{E}_{p_0}[d(f(X),Y)]$ is linear in $q(y|x)$.
		\item \textbf{Feasibility:} There exists at least one mechanism $q(y|x)$ satisfying $\max_i \max_{p \in \Delta(\mathcal{X})_b} I(X_i;Y) \leq L$ (Slater's condition~\cite{citeulike:163662}).
		\item \textbf{Exact Subproblem Solutions:} As in Theorem~\ref{thm:conv-primal}.
		\item \textbf{Boundedness:} The distortion function $d(\cdot,\cdot)$ is bounded above, ensuring the objective is Lipschitz continuous.
	\end{enumerate}
	Then Algorithm~\ref{alg:dual_tradeoff} generates a sequence $\{(p^{(k)}, q^{(k)})\}$ such that:
	\begin{enumerate}
		\item \textbf{Monotonic Improvement:} The expected distortion $D^{(k)} = \mathbb{E}_{p_0}[d(f(X),Y)]$ is non-increasing.
		\item \textbf{Constraint Satisfaction:} The leakage constraint $\max_i I(X_i;Y) \leq L$ is asymptotically satisfied.
		\item \textbf{Stationary Point Convergence:} Every limit point $(p^*, q^*)$ is a KKT point of the dual problem, satisfying:
		\begin{align*}
			q^* &\in \arg\min_{q} \mathbb{E}_{p_0}[d(f(X),Y)] + \lambda^* \cdot (\max_i I_{p^*,q}(X_i;Y) - L) \\
			p^* &\in \arg\max_{p \in \Delta(\mathcal{X})_b} \max_i I_{p,q^*}(X_i;Y) \\
			\lambda^* &\geq 0, \quad \lambda^* \cdot (\max_i I_{p^*,q^*}(X_i;Y) - L) = 0
		\end{align*}
		\item \textbf{Global Convergence for Convex Case:} If the leakage constraint were convex in $q$ for all $p$ (which it is not due to the max over $p$), the algorithm would converge to the global optimum. In our setting, convergence is to a stationary point.
	\end{enumerate}
\end{theorem}

\begin{proof}[Proof Sketch]
	The proof adapts the alternating optimization framework to the constrained minimization structure:
	\begin{enumerate}
		\item \textbf{Penalty Method for Non-Convex Constraints:} The leakage constraint $\max_{p \in \Delta(\mathcal{X})_b} \max_i I(X_i;Y) \leq L$ is non-convex in $q$. We use an exact penalty method with smooth approximation:
		\begin{align*}
			\min_q \mathbb{E}_{p_0}[d(f(X),Y)] + \lambda P_{\text{leakage}}(q)
		\end{align*}
		where $P_{\text{leakage}}(q) = (\text{softplus}(\max_i I(X_i;Y) - L))^2$. This is an exact penalty function: for sufficiently large $\lambda$, its minimizer solves the constrained problem.
		
		\item \textbf{Two-Phase Convergence:}
		\begin{itemize}
			\item \textit{Phase 1 (Outer loop):} For fixed $\lambda$, alternate between:
			\begin{enumerate}
				\item Prior update: $p^{(k+1)} = \arg\max_{p \in \Delta(\mathcal{X})_b} \max_i I_{p,q^{(k)}}(X_i;Y)$
				\item Mechanism update: $q^{(k+1)} = \arg\min_q \mathbb{E}_{p_0}[d] + \lambda P_{\text{leakage}}(q)$
			\end{enumerate}
			By Theorem~\ref{thm:conv-alg1} and convexity of the penalized objective (for fixed $p$), each subproblem converges.
			
			\item \textit{Phase 2 (Penalty update):} Increase $\lambda$ until the constraint violation $|\max_i I(X_i;Y) - L| < \epsilon$. The adaptive update rule (lines 7-11 in Algorithm~\ref{alg:dual_tradeoff}) ensures $\lambda \to \infty$ if needed.
		\end{itemize}
		
		\item \textbf{Convergence to KKT Point:} Let $(p^*, q^*, \lambda^*)$ be a limit point. Then:
		\begin{itemize}
			\item $p^*$ maximizes leakage given $q^*$ (by convergence of Algorithm~\ref{alg:max_leakage})
			\item $q^*$ minimizes the penalized objective (by convergence of Algorithm~\ref{alg:dual_gradient_update_backtracking})
			\item Complementary slackness holds by construction of the penalty update
		\end{itemize}
		
		\item \textbf{Rate Analysis:} The convergence rate is determined by:
		\begin{itemize}
			\item $O(1/k)$ for the outer alternating iterations
			\item Linear convergence for the inner convex subproblems (when strongly convex)
			\item The penalty parameter $\lambda$ may need to increase at a controlled rate to maintain numerical stability
		\end{itemize}
	\end{enumerate}
	
	The non-convexity of the leakage constraint (due to maximization over $p$) prevents global optimality guarantees, but the algorithm converges to a stationary point where no feasible descent direction exists.
\end{proof}

\subsubsection{Practical Implications and Limitations}

While the theoretical guarantees establish a solid foundation for our algorithms, several practical considerations warrant attention. The following points outline key implementation challenges and mitigation strategies.

\begin{enumerate}
	\item \textbf{Local vs. Global Optimality:} Due to non-convexity in Problems~1 and~2, global optimality cannot be guaranteed. We recommend:
	\begin{itemize}
		\item Multiple random initializations to explore the solution space.
		\item For small alphabets, verification via extreme-point enumeration for enumerable case.
	\end{itemize}
	
	\item \textbf{Numerical Stability:} The exponentiated gradient updates require careful handling to avoid underflow. Our implementation employs:
	\begin{itemize}
		\item Log-domain computations for probabilities.
		\item Stabilized softmax operations and clipping of probabilities below $10^{-10}$.
	\end{itemize}
	
	\item \textbf{Convergence Diagnostics:} In practice, we monitor the following metrics to assess convergence:
	\begin{itemize}
		\item Relative objective improvement: $\frac{|I^{(k)} - I^{(k-1)}|}{|I^{(k)}| + \delta}$
		\item Constraint satisfaction: $|H(X) - b|$ for Problem~1, $|I - L|$ for Problem~3
		\item Gradient norms for stationarity verification
	\end{itemize}
	
	\item \textbf{Parameter Sensitivity:} Default parameters ($\eta_0 = 1.0$, $\beta = 0.5$, $\mu = 1.5$) provide robust performance across problem instances. For ill-conditioned problems, adaptive learning rates may be necessary.
\end{enumerate}

These convergence guarantees establish our algorithms as theoretically sound and practically reliable tools for privacy mechanism analysis and design under entropy-constrained adversaries. While global optimality cannot be assured for the non-convex Problems~1 and~2, the alternating optimization framework ensures convergence to practically useful local optima, as demonstrated in our experimental evaluation.

Figure~\ref{fig:flowchart-algorithms} illustrates the dependency relationships among the algorithms and problem formulations introduced in this work.

\begin{figure}[!htb]
	\centering
	\includegraphics[trim={20pt 250pt 20pt 260pt}, clip, width=1\textwidth]{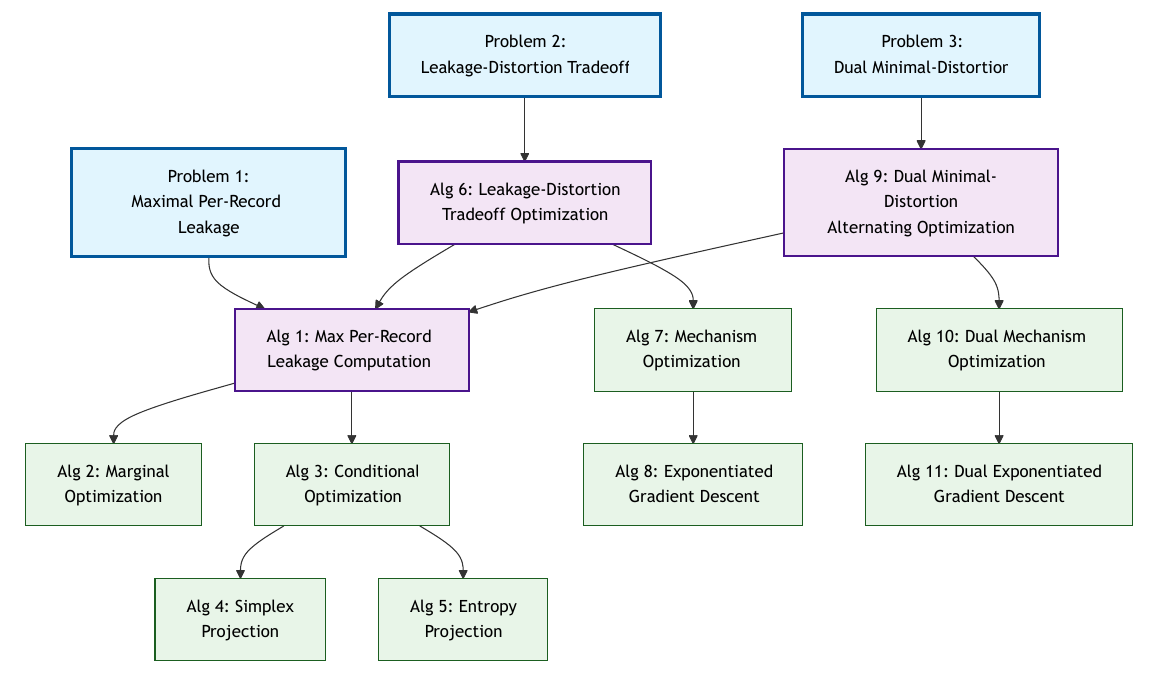}
	\caption{Algorithm Dependency Graph: Core Methods for Maximal Leakage Computation, Leakage-Distortion Optimization, and Dual Formulation with Their Component Subroutines}
	\label{fig:flowchart-algorithms}
\end{figure}

\section{Evaluation}  \label{sec:evaluation}

This section presents a comprehensive evaluation of the proposed information privacy framework and algorithms. We begin by describing the experimental settings, including query functions, leakage bounds, and privacy mechanisms. We then present experimental results for the three core problems: maximal per-record leakage (Problem~1), the primal leakage-distortion tradeoff (Problem~2), and the dual minimal-distortion formulation (Problem~3). 

\subsection{Experimental Settings}
\label{subsec:experimental_settings}

\subsubsection{Query Functions}

We consider two types of query functions to evaluate the proposed algorithms under different scenarios:
\begin{itemize}
	\item \textbf{Linear query:} $f(x) = \sum_{i=1}^{n} x_i \pmod{m}$, where $\mathcal{X}_i = \{0,1\}$ and $\mathcal{Y} = \{0, \ldots, m-1\}$, with sensitivity $\Delta f = 1$. This function represents a simple modular sum, commonly used in privacy-preserving data analysis.
	
	\item \textbf{Non-linear query:} $f(x) = \sum_{1 \leq i < j \leq n} x_i x_j$, where $\mathcal{X}_i = \{0,1\}$ and $\mathcal{Y} = \{0, 1, \ldots, \binom{n}{2}\}$, with sensitivity $\Delta f = n-1$. This quadratic function introduces higher sensitivity and more complex output structure, testing the algorithms under more challenging conditions.
\end{itemize}
For simplicity, we will only consider the query function $f(x) = \sum_{i=1}^{n} x_i \pmod{2}$ in many cases. Other query functions are left as a future work.

\subsubsection{Leakage Bound $b$}
The leakage bound $b$ (entropy constraint) is varied in the range $[0, \log |\mathcal{X}|]$, where $\log |\mathcal{X}|$ is the maximum entropy of the dataset. This allows us to examine the effect of the adversary's prior knowledge on information leakage and utility.

\subsubsection{Utility and Distortion}

We use the expected distortion which is computed as:
\begin{align}
	\mathbb E_{p_0} d(f(X),Y) = \sum_{x\in\mathcal X, y\in\mathcal Y} p_0(x)q(y|x)d(f(x),y),
\end{align}
where $p_0(x)$ is the uniform distribution over $\mathcal X$ and the distortion function is defined as:
\begin{align}
	d(f(x),y) = \begin{cases}
		0 & \text{ if } f(x) = y \\
		|y-f(x)| & \text{ if } f(x) \ne y
	\end{cases}
\end{align}

\subsubsection{Privacy Mechanisms}

We evaluate the following privacy mechanisms to compare their performance under the information privacy framework:

\begin{itemize}
	\item \textbf{(Generalized) Binary-symmetric privacy channel:} A simple channel that, when inputting $x$, flips the output $f(x)$ with probability $p$ and with error probability $p/(m-1)$ for each of remaining symbols: for each $x\in\mathcal X$ and $y\in \mathcal{Y}=\{0,\ldots,m-1\}$,
	\begin{align}
		q(y|x) = \begin{cases}
			1-p  & \text{ if } y = f(x) \\
			\frac{p}{m-1}    & \text{ if } y \ne f(x)
		\end{cases}
	\end{align}
	where $0 \le p \le 1/2$ represents the probability of a bit-flip error. When $m=2$, $\mathcal Y=\{0,1\}$ and
	\begin{align}
		q(y|x) = \begin{cases}
			1-p  & \text{ if } y = f(x) \\
			p    & \text{ if } y \ne f(x)
		\end{cases}
	\end{align}	
	\item \textbf{Laplace privacy mechanism:} Adds Laplace noise to achieve $\epsilon$-differential privacy, with post-processing for discrete outputs.
	\item \textbf{Exponential mechanism:} Samples outputs based on an exponential scoring function, providing $\epsilon$-differential privacy.
\end{itemize}

These mechanisms are chosen to represent both classical differential privacy approaches and the proposed information privacy channels.

\begin{figure}[!bth]
	\centering
	\includegraphics[width=1\textwidth]{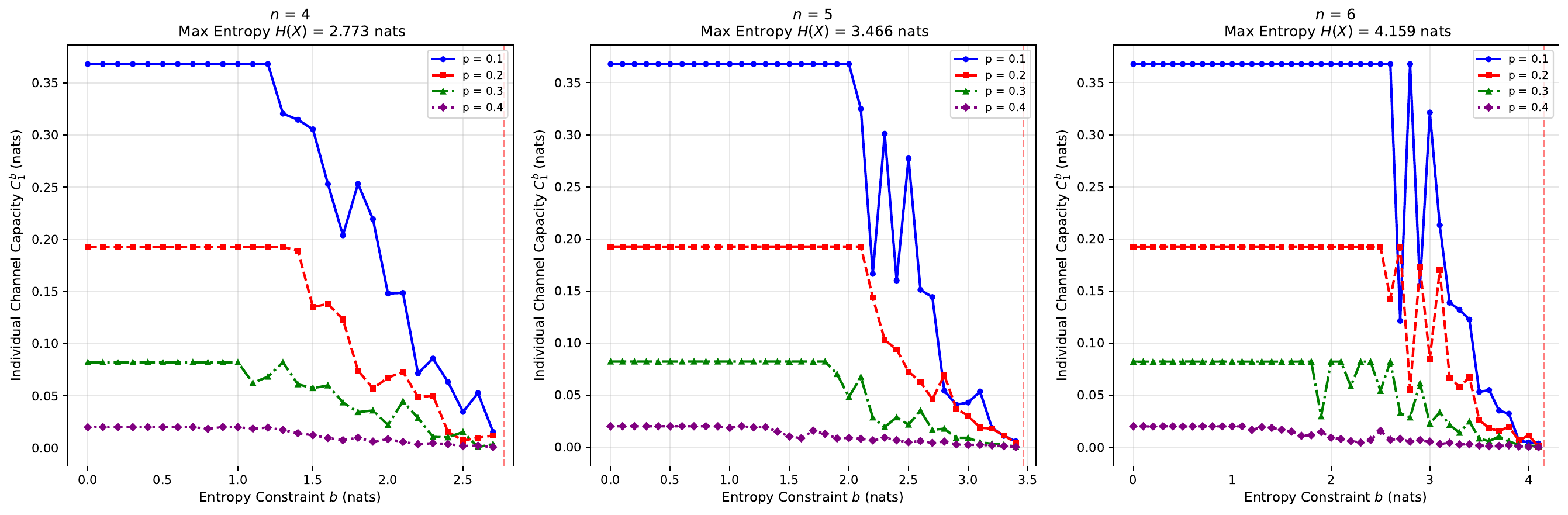}
	\caption{Variation of individual channel capacity $C_1^b$ with entropy constraint $b$, consisting of three subfigures corresponding to different numbers of records $n$ in the dataset: the left subfigure is for $n=4$, the middle subfigure for $n=5$, and the right subfigure for $n=6$. In each subfigure, there are four curves, which correspond to the generalized binary-symmetric privacy channel with different error probabilities $p$ as follows: the curve with the highest initial capacity $C_1^b$ corresponds to $p=0.1$, the next to $p=0.2$, then $p=0.3$, and the curve with the lowest initial capacity $C_1^b$ corresponds to $p=0.4$. 	
		The $x$-axis of each subfigure represents the entropy constraint $b$ (quantifying the lower bound of the adversary's knowledge uncertainty about the dataset, i.e., $H(X) \geq b$), and the $y$-axis denotes the individual channel capacity $C_1^b$ (measuring the maximum amount of information about any individual that an adversary can infer from the privacy channel output). All curves across the three subfigures consistently show that $C_1^b$ decreases as $b$ increases, confirming that the entropy constraint $H(X) \ge b$ effectively limits potential information leakage about any individual by constraining the adversary's prior knowledge. Additionally, for the same $n$ and $b$, a larger $p$ leads to a smaller $C_1^b$, as higher randomness injection further reduces the adversary's ability to infer individual information. The non-monotonic segments observed in all curves originate from the non-convex nature of the $C_1^b$ calculation optimization problem, where our algorithms converge to local optima temporarily during the alternating optimization process.}
	\label{fig:leakage-entropy-constraint}
\end{figure}

\subsection{Experiments for Maximal Per-Record Leakage (Problem~1)}
\label{subsec:experiments-problem1}

We first evaluate the computation of maximal per-record leakage under entropy constraints, as defined in Problem~1. The goal is to understand how the entropy constraint $b$ affects the worst-case information leakage about any individual record.

Figure \ref{fig:leakage-entropy-constraint} illustrates how the individual channel capacity $C_1^b$ varies with the entropy constraint $b$ for different dataset sizes ($n=4,5,6$) and different error probabilities ($p=0.1,0.2,0.3,0.4$). The results demonstrate that $C_1^b$ decreases as $b$ increases across all configurations, confirming that the entropy constraint $H(X) \ge b$ effectively limits potential information leakage about any individual by constraining the adversary's prior knowledge. Additionally, for the same $n$ and $b$, a larger $p$ leads to a smaller $C_1^b$, as higher randomness injection further reduces the adversary's ability to infer individual information. The non-monotonic segments observed in all curves originate from the non-convex nature of the $C_1^b$ calculation optimization problem, where our algorithms converge to local optima temporarily during the alternating optimization process.

The results confirm that $C_1^b$ decreases as $b$ increases, demonstrating that the entropy constraint effectively limits the adversary's ability to infer individual information. Additionally, higher noise levels (larger $p$) lead to lower leakage, as expected.

\begin{figure}[!tbh]
	\centering
	\includegraphics[width=1\textwidth]{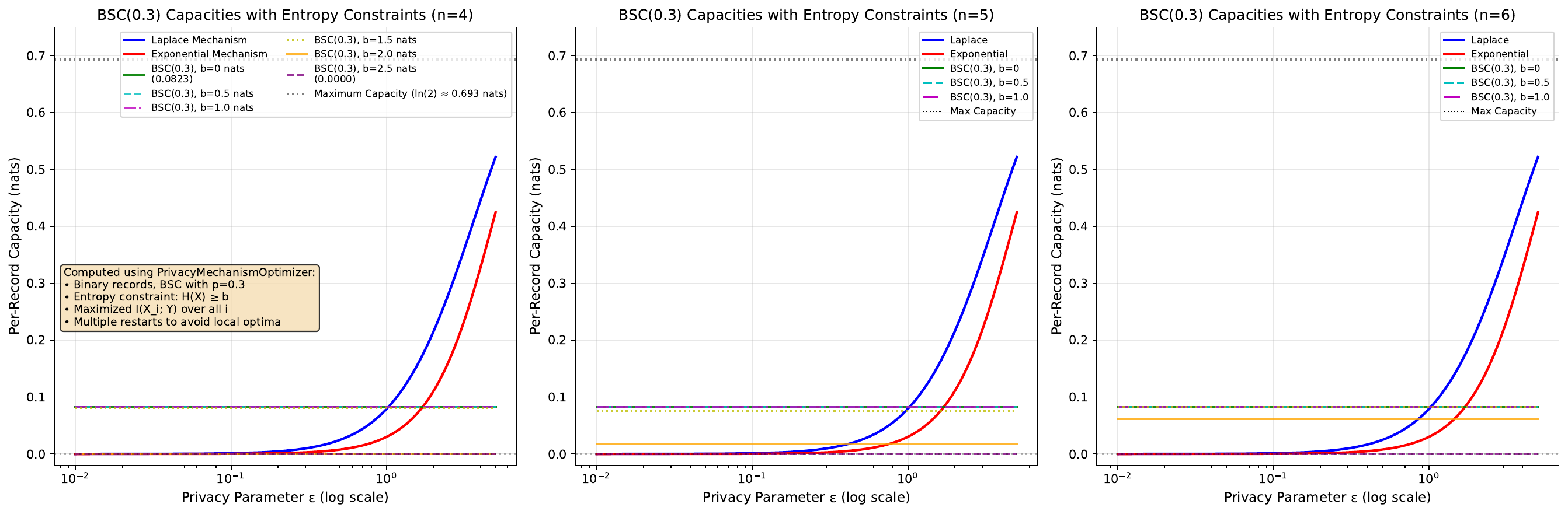}
	\caption{Comparison of per-record capacity $C_1^b$ (maximum mutual information $I(X_i;Y)$ in nats) across three privacy mechanisms for different dataset sizes. From left to right: $n=4$, $n=5$, and $n=6$ records. In each subfigure, three mechanisms are shown: Laplace mechanism (blue curve), Exponential mechanism (red curve), and binary symmetric channel with crossover probability $p=0.3$ (BSC(0.3), shown as horizontal lines). The $x$-axis represents the privacy parameter $\epsilon$ (log scale) for the Laplace and Exponential mechanisms, while the $y$-axis denotes the per-record capacity. The Laplace and Exponential mechanisms exhibit the standard differential privacy trade-off where capacity decreases as $\epsilon$ increases (i.e., stronger privacy). For BSC(0.3), which does not depend on $\epsilon$, we show capacities under different entropy constraints $b = 0, 0.5, 1.0, 1.5, 2.0, 2.5$ nats (green to purple lines, with decreasing line thickness and varying styles). Each BSC(0.3) capacity value is computed by solving the optimization problem of maximizing $I(X_i;Y)$ subject to the entropy constraint $H(X) \geq b$ for the corresponding dataset size $n$. The results demonstrate that for a fixed $\epsilon$, the Laplace and Exponential mechanisms achieve higher capacity than BSC(0.3) with $b>0$, but when $b=0$ (no entropy constraint) BSC(0.3) has constant capacity $C_1^0 \approx 0.0931$ nats for $n=4$, decreasing for larger $n$. As $b$ increases, the capacity of BSC(0.3) decreases, reflecting the limitation imposed by the adversary's uncertainty about the dataset. The horizontal dotted line at $\ln(2) \approx 0.693$ nats represents the theoretical maximum capacity for any binary channel. The non-monotonic segments in some curves originate from the non-convex nature of the $C_1^b$ optimization problem, where the alternating algorithm may converge to local optima despite multiple random restarts.}
	\label{fig:bsc-laplace-exp-comparison}
\end{figure}

\begin{figure}[!hbt]
	\centering
	\includegraphics[width=1\textwidth]{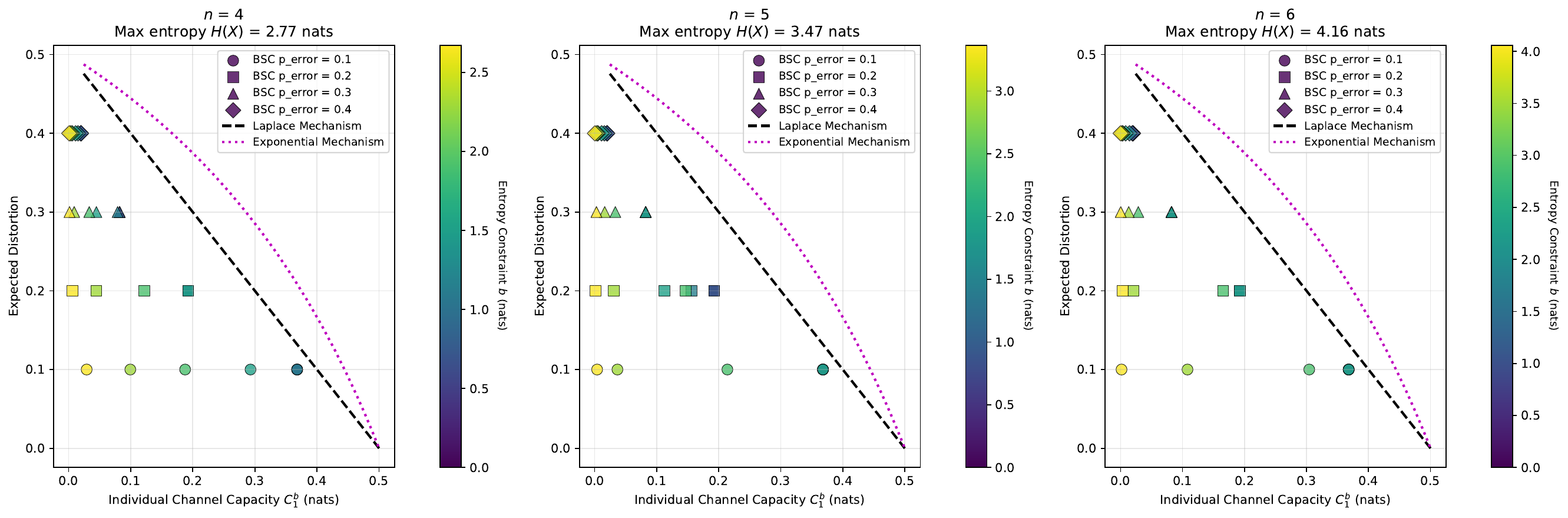}
	\caption{Comparison of individual channel capacity $C_1^b$ versus expected distortion $\mathbb E_{p_0} d(f(X),Y)$ across three privacy mechanisms for the binary parity query function, consisting of three subfigures corresponding to different numbers of records $n$ in the dataset: the left subfigure is for $n=4$ (with maximum entropy of 2.77 nats), the middle subfigure for $n=5$ (with maximum entropy of 3.47 nats), and the right subfigure for $n=6$ (with maximum entropy of 4.159 nats). In each subfigure, there are three kinds of curves, which correspond to different privacy mechanisms as follows: the four dotted curves with the lower privacy leakages at the same distortion level corresponds to the binary symmetric privacy channel, the next curve corresponds to the Laplace mechanism, and the curve with the highest privacy leakage corresponds to the Exponential mechanism. 
		The $x$-axis denotes the maximal privacy leakage/privacy budget, measured by the privacy parameter $\epsilon$: for the binary symmetric privacy channel, $\epsilon=C_1^b$ to ensure it satisfies $\epsilon$-Information Privacy with respect to the adversary class $\mathbb{P}_b$; for fair comparison, the privacy budget of the Laplace mechanism and Exponential mechanism  is set equal to $\epsilon$ to guarantee $\epsilon$-Differential Privacy.
		The $y$-axis of each subfigure represents the expected distortion $\mathbb E_{p_0} d(f(X),Y) $.
		All curves across the three subfigures consistently show two key results: (1) Under the same expected distortion, the binary symmetric privacy channel requires a smaller privacy budget $\epsilon$ (i.e., lower information leakage) than both the Laplace mechanism and the Exponential mechanism; (2) Increasing the entropy constraint $b$ further reduces the privacy leakage $C_1^b$ of the binary symmetric privacy channel at the same distortion level, which highlights the utility advantage of the information privacy framework enabled by the bounded knowledge assumption ($H(X) \geq b$).}
	
	\label{fig:leakage-distortion}
\end{figure}

\subsection{Experiments for Leakage-Distortion Tradeoff (Problem~2)}
\label{subsec:experiments-problem2}

Next, we evaluate the primal leakage-distortion tradeoff, which aims to design privacy mechanisms that minimize worst-case leakage while satisfying a distortion constraint.

\subsubsection{Distortion Comparison to Differential Privacy Mechanisms}

We now compare the utility-privacy tradeoffs of our binary symmetric privacy channel against two fundamental differential privacy mechanisms: Laplace and Exponential mechanisms. For equitable comparison, we set the privacy parameter $\epsilon$ for each DP mechanism equal to the information leakage $C_1^b$ of our channel, ensuring equivalent privacy guarantees.

\begin{figure}[!hbt]
	\centering
	\includegraphics[width=1\textwidth]{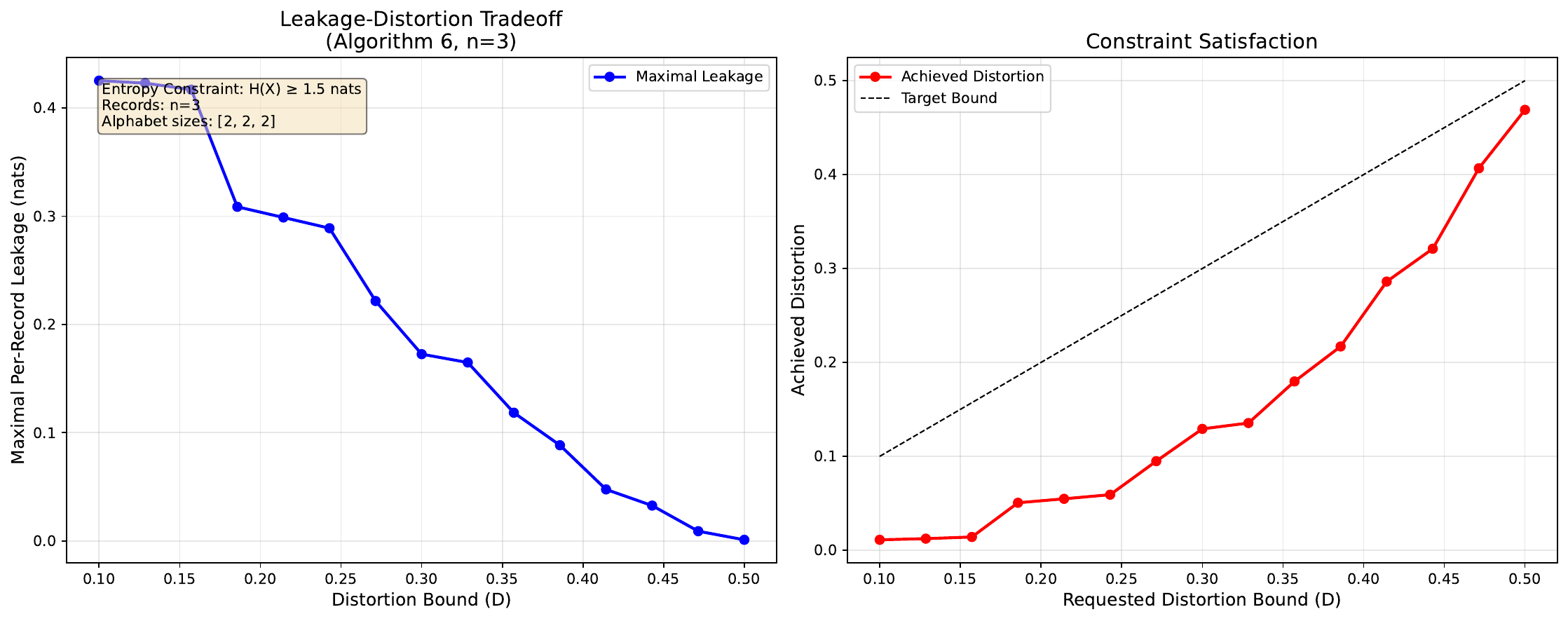}
	\caption{
		\textbf{Leakage-Distortion Tradeoff Analysis for Binary Symmetric Privacy Channel for $b =1.5 \text{ nats}$ and $n=3$.} 
		(Left) The fundamental privacy-utility tradeoff curve showing maximal per-record leakage $C_1^b$ versus distortion bound $D$. As the distortion constraint $D$ increases (moving rightward), the mechanism can inject more noise, resulting in decreased information leakage about individual records. This demonstrates the inherent tension between data utility (low distortion) and privacy protection (low leakage). (Right) Constraint satisfaction analysis showing achieved distortion versus requested distortion bound $D$. The red curve (achieved distortion) closely tracks but remains below the black dashed line (target bound), indicating that the optimization algorithm successfully finds mechanisms that satisfy the distortion constraint while minimizing information leakage. The slight conservatism (achieved distortion < requested bound) ensures robust constraint satisfaction while maintaining near-optimal privacy protection. Together, these plots validate Algorithm \ref{alg:primal_tradeoff}'s ability to compute the leakage-distortion Pareto frontier for entropy-constrained adversaries. 
	}
	\label{fig:leakage-distortion-tradeoff}
\end{figure}

The detailed evaluation of the distortions of the binary-symmetric privacy channel, the Laplace mechanism and the Exponential mechanism of differential privacy are presented in Section \ref{subsec:capacity-distortion-binary}.

\paragraph{Distortion Comparison}
Figure \ref{fig:leakage-distortion} compares the individual channel capacity $C_1^b$ versus expected distortion across the three privacy mechanisms for different dataset sizes ($n=4,5,6$). The results demonstrate two key findings: (1) Under the same expected distortion, the binary symmetric privacy channel requires a smaller privacy budget $\epsilon$ (i.e., lower information leakage) than both the Laplace and Exponential mechanisms; (2) Increasing the entropy constraint $b$ further reduces the privacy leakage $C_1^b$ of the binary symmetric privacy channel at the same distortion level, which highlights the utility advantage of the information privacy framework enabled by the bounded knowledge assumption ($H(X) \geq b$).

\subsubsection{Additional Experiments}
We extend the evaluation to the fundamental privacy-utility tradeoff curve showing maximal per-record leakage $C_1^b$ versus distortion bound $D$. We need to test, as the distortion constraint $D$ increases, the mechanism can inject more noise, resulting in decreased information leakage about individual records. If this is true, it demonstrates the inherent tension between data utility (low distortion) and privacy protection (low leakage).
Figure \ref{fig:leakage-distortion-tradeoff} plots the above leakage-distortion curves for $b =1.5 \text{ nats}$.

\subsection{Experiments for Dual Minimal-Distortion Formulation (Problem~3)}

Finally, we evaluate the dual minimal-distortion formulation, which minimizes expected distortion under a leakage constraint.

This section presents experimental results for Algorithm~\ref{alg:dual_tradeoff} on the parity function  with $b=0$ and varying leakage bound $L$. The algorithm successfully finds mechanisms that satisfy the leakage constraint while minimizing distortion, though instability in the penalty term leads to zigzag curves in the optimization trajectory.

\begin{figure}[!hbt]
	\centering
	\includegraphics[width=1\textwidth]{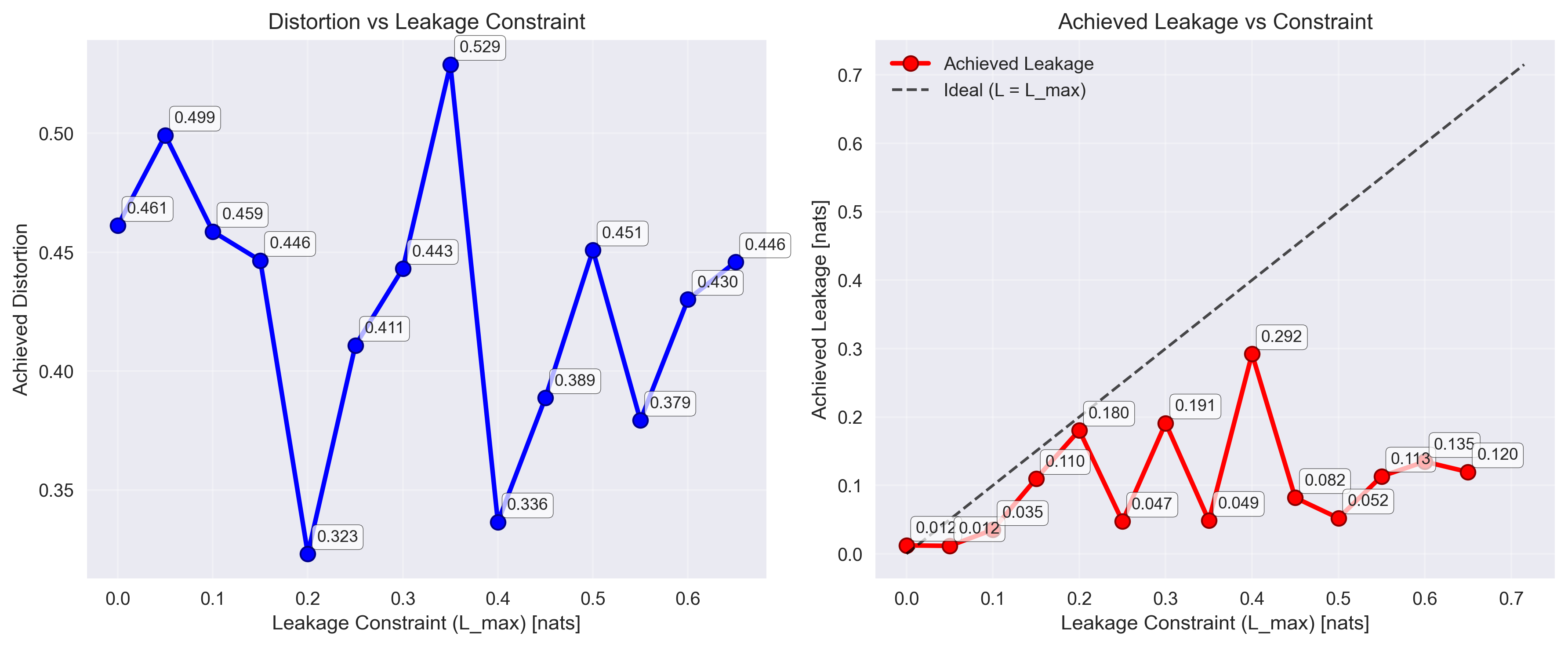}
	\caption{
		\textbf{Minimal distortion for leakage constraints for Binary Symmetric Privacy Channel with $n=2$ and $b=0$.} 
		(Left) The fundamental privacy-utility tradeoff curve showing minimal distortion versus leakage bounded $L$. As the leakage constraint $L$ increases (moving rightward), the mechanism can inject less noise, resulting in decreased outputs distortion. This demonstrates the inherent tension between data utility (low distortion) and privacy protection (low leakage). (Right) Constraint satisfaction analysis showing achieved leakage versus requested leakage bound $L$. The red curve (achieved leakage) closely tracks but remains below the black dashed line (target bound), indicating that the optimization algorithm successfully finds mechanisms that satisfy the leakage constraint while minimizing distortion. The slight conservatism (achieved leakage < requested bound) ensures robust constraint satisfaction while maintaining near-optimal privacy protection.  
	}
	\label{fig:merged_distortion_leakage_bound_tradeoff}
\end{figure}

\subsection{Discussion}

The experimental results consistently demonstrate the advantages of the information privacy framework over classical differential privacy mechanisms. By incorporating an entropy constraint on the adversary's prior knowledge, the framework enables more nuanced privacy-utility tradeoffs, often achieving lower leakage for the same distortion or lower distortion for the same leakage. The proposed algorithms effectively solve the three core problems, though challenges remain in handling non-convexities and high-dimensional spaces.

Future work will focus on scaling the algorithms to larger datasets, improving convergence stability for the dual formulation, and extending the framework to more complex query functions and adversary models.

\section{Related Work}
\label{sec:related_works}

Our work sits at the intersection of information theory, data privacy, and mechanism design. This section surveys the relevant literature, highlighting the foundational concepts we build upon and clarifying the novel contributions of our approach.

\subsection{Foundations in Information Theory}

The three core problems formulated in this paper---maximal per-record leakage, the primal leakage-distortion tradeoff, and its dual distortion-minimization problem---are conceptually and formally inspired by fundamental problems in information theory \cite{DBLP:books/daglib/0016881}. However, they introduce critical complexities tailored to the privacy setting.

\textbf{Conceptual and Methodological Similarities:} We deliberately adopt the language and formalism of information theory. Our objectives (maximizing or minimizing mutual information), constraints (distortion bounds), and algorithmic solutions (alternating optimization, convex analysis) are direct generalizations of classic information-theoretic frameworks. Specifically:
\begin{itemize}
	\item \textbf{Problem 1 (Maximal Leakage)} generalizes the \textbf{channel capacity} problem, which seeks the maximum mutual information $I(X;Y)$ over input distributions.
	\item \textbf{Problem 2 (Primal Tradeoff)} generalizes the \textbf{rate-distortion} function, which characterizes the minimum achievable rate (information) for a given fidelity constraint.
	\item \textbf{Problem 3 (Dual Formulation)} generalizes the \textbf{dual of the rate-distortion} problem, minimizing distortion for a given rate constraint.
\end{itemize}

\textbf{Key Distinctions and Added Complexities:} The privacy context renders these problems significantly more challenging than their information-theoretic analogues due to two major factors:
\begin{enumerate}
	\item \textbf{High-Dimensionality:} Instead of optimizing a single mutual information $I(X;Y)$ for a monolithic source, we must consider the \textit{per-record} leakage $I(X_i;Y)$ for all $i \in [n]$. This scales the problem with the number of records $n$.
	\item \textbf{The Entropy Constraint:} Replacing Differential Privacy's independence assumption with the entropy constraint $H(X) \geq b$ on the adversary's prior is more realistic but adds a non-trivial, global nonlinear constraint to the optimization. This constraint is central to modeling an adversary with bounded knowledge.
\end{enumerate}

In essence, our work elevates the well-understood paradigms of channel capacity and rate-distortion to a more challenging and realistic setting for data privacy, where the threat model explicitly accounts for an adversary with limited prior knowledge.
Table \ref{tab:comparison_information_theory} summarizes the above comparisons.

\begin{table}[!htb]
	\centering\tiny
	\caption{Correspondence between Information-Theoretic Problems and Our Privacy Formulations}
	\label{tab:comparison_information_theory}
	\begin{tabular}{p{1.3cm} p{2.7cm} p{3.9cm} p{2.5cm} p{3.5cm}}
		\toprule
		\textbf{Feature} & \textbf{Classic Information Theory Problem} & \textbf{Our Privacy Problem} & \textbf{Key Similarities} & \textbf{Key Differences} \\
		\midrule
		\textbf{Problem 1} & \textbf{Channel Capacity:} $C = \max_{p(x)} I(X;Y)$ & \textbf{Maximal Per-Record Leakage:} $\mathcal{L}(b) = \max\limits_{i} \max\limits_{\substack{p(x):\\ H(X)\geq b}} I(X_i;Y)$ & \begin{itemize}[leftmargin=*, nosep, topsep=0pt]\item Optimization form \item Alternating/Blahut-Arimoto algorithms \end{itemize} & \begin{itemize}[leftmargin=*, nosep, topsep=0pt]\item \textbf{Dimensionality:} Scales with $n$ records \item \textbf{Constraint:} $H(X) \geq b$ \end{itemize} \\
		\midrule
		\textbf{Problem 2} & \textbf{Rate-Distortion:} $R(D) = \min\limits_{\substack{q(y|x):\\ \mathbb{E}[d]\leq D}} I(X;Y)$ & \textbf{Primal Leakage-Distortion:} $\epsilon(D,b) = \min\limits_{\substack{q(y|x):\\ \mathbb{E}_{p_0}[d]\leq D}} \max\limits_{i} \max\limits_{\substack{p(x):\\ H(X)\geq b}} I(X_i;Y)$ & \begin{itemize}[leftmargin=*, nosep, topsep=0pt]\item Trade-off structure \item Minimization over channels \end{itemize} & \begin{itemize}[leftmargin=*, nosep, topsep=0pt]\item \textbf{Objective:} Worst-case per-record info \item \textbf{Distributions:} Distinguishes $p_0(x)$ from $p(x)$ \item \textbf{Constraint:} $H(X) \geq b$ \end{itemize} \\
		\midrule
		\textbf{Problem 3} & \textbf{Dual of Rate-Distortion:} $D(R) = \min\limits_{q(y|x)} \mathbb{E}[d] \text{ s.t. } I(X;Y) \leq R$ & \textbf{Dual Minimal-Distortion:} 
		$\begin{array}{l}
			\mathcal D(L,b) = \displaystyle\min_{q(y|x)} \mathbb{E}_{p_0}[d(f(X),Y)] \\
			\text{s.t. }\max\limits_{i} \max\limits_{\substack{p(x):\\ H(X)\geq b}} I(X_i;Y) \leq L
		\end{array}$ & \begin{itemize}[leftmargin=*, nosep, topsep=0pt]\item Dual perspective \item Convexity properties \end{itemize} & \begin{itemize}[leftmargin=*, nosep, topsep=0pt]\item \textbf{Constraint:} Complex worst-case leakage bound \item \textbf{Distributions:} Distinguishes $p_0(x)$ from $p(x)$ \end{itemize} \\
		\bottomrule
	\end{tabular}
\end{table}

Prior information-theoretic analyses of differential privacy (summarized in \cite{WU2026123035}) often lack this seamless connection to the optimization frameworks of capacity and rate-distortion. We argue this is due to the structure of DP constraints: the information privacy framework uses a single, tractable inequality $H(X) \geq b$, whereas differential privacy relies on many pairwise ratio constraints $\log \frac{q(y|x)}{q(y|x')} \leq \epsilon$, which complicate the formulation of unified optimization problems.

\subsection{Utility-Privacy Tradeoffs in Differential Privacy}

A significant body of work focuses on designing optimal differentially private mechanisms or deriving theoretical bounds on their utility. While these efforts have provided deep insights, the prevailing practice of adding noise (e.g., Laplace) to achieve DP often obscures the underlying information-theoretic nature of the privacy-utility tradeoff. This makes it challenging to obtain tight analytical characterizations for complex query functions and application scenarios. Our information privacy (IP) framework complements DP by providing a unified optimization-based perspective, enabling the derivation of precise tradeoffs---such as Problems \ref{def:max-leakage}--\ref{def:dual-tradeoff}---that are largely intractable under the standard DP constraint structure.

\begin{itemize}
	\item \textbf{Optimal Mechanism Design:} A key line of research seeks the \textit{optimal} DP mechanism for a given utility metric. For local differential privacy, Kairouz et al. \cite{DBLP:conf/nips/KairouzOV14} introduced \textit{staircase mechanisms} and proved their optimality for a class of information-theoretic utilities, framing the search as a linear program. Geng and Viswanath \cite{DBLP:journals/tit/GengV16a, DBLP:journals/tit/GengV16} derived optimal noise-adding mechanisms for approximate DP, while Ghosh et al. \cite{DBLP:conf/stoc/GhoshRS09} explored universally utility-maximizing mechanisms.
	\item \textbf{Theoretical Utility Bounds:} Other works establish mini-max lower bounds on the distortion (error) incurred by any DP mechanism for specific query classes, such as linear queries \cite{DBLP:conf/tcc/De12, DBLP:conf/focs/DuchiJW13, DBLP:conf/stoc/NikolovTZ13, DBLP:journals/jpc/SteinkeU16}, range counting \cite{DBLP:conf/stoc/MuthukrishnanN12}, or in high-dimensional settings \cite{10.1007/978-3-032-12290-2_11}. The geometry of DP error is explored in \cite{DBLP:conf/stoc/HardtT10}.
	\item \textbf{Workload-Aware Mechanisms:} The \textit{matrix mechanism} \cite{DBLP:conf/icdt/LiM13, DBLP:journals/vldb/LiMHMR15} demonstrates that adding correlated noise tailored to a set of linear queries (a workload) can achieve lower error than independent noise addition.
	\item \textbf{Composition and Its Limits:} The composition theorems of DP \cite{DBLP:journals/tit/KairouzOV17, DBLP:conf/nips/Lyu22, DBLP:conf/stoc/Vadhan023, DBLP:journals/ieiceta/SaitoCT25} are fundamental for building complex algorithms, and their tightness directly impacts utility analyses for sequential applications.
\end{itemize}

While these works deeply explore the utility frontier of DP, they operate within the independence-based adversary model. Our work transitions to a bounded-knowledge adversary model, which alters the fundamental privacy-utility tradeoff landscape.

\subsection{Information Leakage in Privacy and Cryptography}

Quantifying information leakage is a central concern in both privacy and cryptography.

\begin{itemize}
	\item \textbf{Information Leakage of DP Mechanisms:} Although DP mechanisms are inherently channels, explicit calculation of their per-record mutual information leakage has been less common, likely due to the predominance of the modular composition paradigm. Works like that of Alvim et al. \cite{DBLP:conf/ifip1-7/AlvimACDP11} begin to bridge this gap by analyzing the trade-off between utility and information leakage within DP. Our evaluations in Sections \ref{subsec:experiments-problem1} and \ref{subsec:experiments-problem2} contribute to this direction by empirically computing the per-record capacity of Laplace and Exponential mechanisms for some specific functions.
	\item \textbf{Operational Measures of Leakage:} In a broader context, significant research formalizes information leakage using concepts like \textit{guessing entropy} or \textit{min-entropy leakage}. Issa et al. \cite{DBLP:journals/tit/IssaWK20} and Kurri et al. \cite{DBLP:journals/tit/KurriSK24} propose operational frameworks for measuring leakage, often focusing on scenarios with a single secret ($n=1$). These are closely related to the local information privacy setting.
	\item \textbf{Leakage-Resilient Cryptography:} This field \cite{DBLP:journals/tac/DuanXYT24} formally models an adversary that can learn \textit{bounded} functions of internal secret state (e.g., via side channels). The security proofs explicitly account for this bounded leakage, sharing a philosophical similarity with our entropy-constrained adversary model, though applied in a different cryptographic context.
\end{itemize}

Our work unifies aspects of these perspectives by employing Shannon mutual information as the leakage measure---a choice that enables direct connection to rate-distortion theory---while imposing an entropy constraint to model a realistic, bounded-knowledge adversary.

\subsection{Convergence Proof Methods}

Our Algorithm \ref{alg:max_leakage} follows an alternating optimization scheme similar to the Generalized Alternating Minimization (GAM) framework~\cite{JMLR:v6:gunawardana05a}. While modern Kurdyka-\L{}ojasiewicz (KL) theory could also be applied, the GAM framework provides a more direct and tailored convergence analysis for our problem structure. GAM is highly suitable for proving convergence of Algorithm~\ref{alg:max_leakage}; our Theorem~\ref{thm:conv-alg1} already follows the GAM structure and is correct under the stated assumptions.

\section{Conclusion} 
\label{sec:conclusion}

This paper addresses the fundamental challenge of quantifying and optimizing privacy-utility tradeoffs under realistic adversary assumptions. By moving beyond the independence assumption of differential privacy to entropy-constrained adversaries, we develop a more nuanced framework for privacy protection that better reflects real-world scenarios where adversaries have partial but bounded knowledge.

Our main contributions are threefold. First, we formalize three core privacy optimization problems: maximal per-record leakage (Problem~1), the primal leakage-distortion tradeoff (Problem~2), and its dual minimal-distortion formulation (Problem~3). These problems establish direct connections to classical information theory, extending channel capacity, rate-distortion, and dual rate-distortion concepts to the privacy domain.

Second, we develop efficient alternating optimization algorithms with provable convergence guarantees. Algorithm~\ref{alg:max_leakage} solves the maximal leakage problem by exploiting the convex-concave structure of mutual information. Algorithm~\ref{alg:primal_tradeoff} addresses the leakage-distortion tradeoff through alternating mechanism and adversary updates. Algorithm~\ref{alg:dual_tradeoff} provides globally convergent solutions to the dual distortion minimization problem. Our convergence analysis establishes local convergence for Problems~1--3.

Third, our experimental evaluation validates the proposed algorithms and demonstrates their practical utility. On binary symmetric channels, our mechanisms achieve favorable privacy-utility tradeoffs compared to classical differential privacy approaches. The results show that incorporating entropy constraints enables more precise privacy protection without sacrificing utility.

\paragraph{Limitations and Future Work.} While our framework advances the state of privacy quantification, several limitations warrant further investigation. The computational complexity grows exponentially with dataset size, necessitating approximation techniques for large-scale applications. Future work should explore:
\begin{itemize}
	\item Efficient approximations for high-dimensional settings
	\item Extension to continuous data domains
	\item Integration with machine learning pipelines
	\item Adaptive mechanisms that adjust to evolving adversary knowledge
\end{itemize}

\paragraph{Broader Impact.} This work provides both theoretical foundations and practical tools for privacy risk assessment and mechanism design. By enabling precise quantification of information leakage under realistic adversary models, our framework supports regulatory compliance, privacy-preserving data sharing, and trustworthy AI development. The algorithms and insights presented here contribute to building a more privacy-aware digital ecosystem where data utility and individual protection can coexist harmoniously.

\section*{Code Availability}

The Python implementations of Algorithms~\ref{alg:max_leakage}--\ref{alg:dual_gradient_update_backtracking} presented in this paper are publicly available to support reproducibility and further research. The complete source code, including all implementation files, test scripts, and documentation, can be accessed from the GitHub repository:
\url{https://github.com/EpsilonOmega2024/privacy-mechanism-optimization}.
Readers are encouraged to download, use, and extend the code under the terms of the MIT License.

\section*{Funding}
This work was supported by the following grants: 
the National Natural Science Foundation of China (NSFC) [62566035],
the Gansu Provincial Science and Technology Major Project - Industrial Category [22ZD6GA047],
the Gansu Provincial Higher Education Youth Doctoral Support Program [2023QB-070], 
and the University-level Research Project of Lanzhou University of Finance and Economics [Lzufe2021B-014], 
for the projects ``Graph-Embedded Clustering Algorithms for High-Dimensional Data and Their Applications",
``Research on the Application of Artificial Intelligence in Tumor Imaging Diagnosis and Treatment",
``Research on Key Issues in Information-Theoretic Foundations of Data Privacy", 
and ``Research on Key Issues in Privacy-Preserving Federated Learning", respectively.


\bibliographystyle{plain}

	\section{Appendix}

	\subsection{Foundational Derivatives}
	\label{sec:foundational_derivatives}
	
	This section presents partial derivatives of information-theoretic quantities essential for our optimization framework. All results assume discrete random variables with finite alphabets.
	
	\paragraph{Notation}
	Let $X = (X_1,\dots,X_n)$ be a random vector with joint distribution $p(x)$ over $\mathcal{X} = \prod_{i=1}^n \mathcal{X}_i$, and $Y$ be the output of mechanism $q(y|x)$. We denote:
	\begin{itemize}
		\item $p(y|x_i) = \sum_{x_{-i}} p(x_{-i}|x_i) q(y|x)$: Output distribution conditioned on $X_i$
		\item $p(y) = \sum_x p(x) q(y|x)$: Marginal output distribution
	\end{itemize}
	
	\begin{theorem}[Mutual Information Gradients]   \label{thm:mi_gradient}
		For $I(X_i;Y) = D_{\text{KL}}(p(x_i,y) \parallel p(x_i)p(y))$:
		\begin{enumerate}
			\item Derivative w.r.t. channel probability:
			\begin{equation}
				\frac{\partial I(X_i; Y)}{\partial q(y | x)} = p(x) \log \frac{p(y | x_i)}{p(y)}.
				\label{eq:mi_channel}
			\end{equation}
			
			\item Derivative w.r.t. marginal probability:
			\begin{equation}
				\frac{\partial I}{\partial p(a)} = \sum_y p(y | a) \log \frac{p(y | a)}{p(y)}  - 1.
				\label{eq:mi_marginal}
			\end{equation}
			
			\item Derivative w.r.t. conditional distribution:
			\begin{equation}
				\frac{\partial I(X_i; Y)}{\partial p(x_{-i} | x_i)} = p(x_i) +  p(x_i) \sum_y q(y | x) \log \frac{p(y | x_i)}{p(y)} - p(x_i)^2 \sum_y \frac{p(y | x_i) q(y | x)}{p(y)}.
				\label{eq:mi_conditional}
			\end{equation}
		\end{enumerate}
	\end{theorem}
	
	\begin{proof}
		
		\textbf{Proof of Part 1:}
		Express mutual information as:
		\[
		I(X_i; Y) = \sum_{x} p(x) \sum_{y} q(y| x) \log \frac{p(y| x_i)}{p(y)}
		\]
		where $x = (x_i, x_{-i})$. The derivative w.r.t. a specific $q(y'| x')$ is:
		\[
		\frac{\partial I}{\partial q(y'| x')} = p(x') \log \frac{p(y'| x_i')}{p(y')} + \sum_{x} p(x) \sum_{y} q(y| x) \frac{\partial}{\partial q(y'| x')} \left( \log \frac{p(y| x_i)}{p(y)} \right)
		\]
		The second term simplifies to zero because:
		\[
		\frac{\partial}{\partial q(y'| x')} \log \frac{p(y| x_i)}{p(y)} = 
		\delta_{y,y'} \delta_{x_i,x_i'} \frac{p(x_{-i}'| x_i')}{p(y'| x_i')} - \delta_{y,y'} \frac{p(x')}{p(y')},
		\]
		where the Kronecker delta function $\delta_{a,a'}$ activates only when $a = a'$.
		After summation:
		\[
		\sum_{x} p(x) \sum_{y} q(y| x) \left[ \cdots \right] = p(x') - p(x') = 0
		\]
		Thus:
		\[
		\frac{\partial I}{\partial q(y'| x')} = p(x') \log \frac{p(y'| x_i')}{p(y')}
		\]
		Dropping primes gives the general result.
		
		\textbf{Proof of Part 2:}
		Express mutual information as:
		\begin{align}
			I = \sum_{x_i} p\left(x_i\right) \sum_y p\left(y| x_i\right) \log \frac{p\left(y| x_i\right)}{p(y)}.
		\end{align}
		Derivative w.r.t. $p\left(x_i\right)$ at $x_i = a$:
		\begin{align} \label{equation-11}
			\frac{\partial I}{\partial p(a)} = \sum_y p(y| a) \log \frac{p(y| a)}{p(y)} 
			+ p(a) \sum_y p(y| a) \frac{\partial}{\partial p(a)} \left( \log \frac{p(y| a)}{p(y)} \right)  -\frac{\partial}{\partial p(a)} \left(\sum_{x_i \neq a} p(x_i) \sum_y p(y | x_i) \log\frac{p(y | x_i)}{p(y)} \right).
		\end{align}
		Since $\frac{\partial}{\partial p(a)} \left( \log \frac{p(y| a)}{p(y)} \right) = -\frac{p(y| a)}{p(y)}$ and the above third term equals
		\begin{align}
			-\sum_y p(y | a) + p(a) \sum_y \frac{[p(y | a)]^2}{p(y)}, 
		\end{align}
		by substituting the above two results into (\ref{equation-11}) gives the result.
		
		\textbf{Proof of Part 3:}
		Express mutual information as:
		\[
		I = \sum_{x_i} p\left(x_i\right) \sum_y \left[ \sum_{x_{-i}} p\left(x_{-i}| x_i\right) q(y| x) \right] \log \frac{p\left(y| x_i\right)}{p(y)}.
		\]
		Derivative w.r.t. specific $p(x_{-i}^*| x_i^*)$:
		\begin{subequations}  \label{eq:equation-2}
			\begin{align}
				\frac{\partial I}{\partial p(x_{-i}^*| x_i^*)} &= 
				p(x_i^*) \sum_y q(y| x_i^*, x_{-i}^*) \log \frac{p(y| x_i^*)}{p(y)} \\
				&+ p(x_i^*) \sum_y p(y| x_i^*) \frac{\partial}{\partial p(x_{-i}^*| x_i^*)} \left( \log \frac{p(y| x_i^*)}{p(y)} \right).
			\end{align}
		\end{subequations}
		The derivative in the second term is:
		\[
		\frac{\partial}{\partial p(x_{-i}^*| x_i^*)} \left( \log \frac{p(y| x_i^*)}{p(y)} \right) = 
		\frac{q(y| x_i^*, x_{-i}^*)}{p(y| x_i^*)} - \frac{p(x_i^*) q(y| x_i^*, x_{-i}^*)}{p(y)}.
		\]
		By substitution and simplification, the second term becomes 
		\begin{align}
			p(x_i^*) - p(x_i^*)^2 \sum_y\frac{ p(y| x_i^*) q(y| x)}{p(y)}.
		\end{align}
		We get the final result by substituting the second term of \eqref{eq:equation-2} with the above result.
	\end{proof}
	
	\begin{theorem}[Entropy Gradient]  \label{thm:entropy_gradient}
		The entropy $H(X) = -\sum_{x} p(x) \log p(x)$ has derivatives:
		\begin{enumerate}
			\item W.r.t. joint probability:
			\[
			\frac{\partial H(X)}{\partial p(x)} = -\log p(x) - 1
			\]
			
			\item W.r.t. marginal probability:
			\[
			\frac{\partial H(X)}{\partial p\left(x_i\right)} = H\left(X_{-i} | x_i\right) - \log p\left(x_i\right) - 1
			\]
			
			\item W.r.t. conditional distribution:
			\[
			\frac{\partial H(X)}{\partial p\left(x_{-i} | x_i\right)} = -p\left(x_i\right) \left( \log \left[ p\left(x_i\right) p\left(x_{-i} | x_i\right) \right] + 1 \right)
			\]
		\end{enumerate}
	\end{theorem}

	\begin{proof}
		
		\textbf{Proof of Part 1:}
		The entropy is defined as:
		\[
		H(X) = -\sum_{x'} p(x') \log p(x')
		\]
		Taking the derivative with respect to a specific probability mass $p(x)$:
		\[
		\frac{\partial H(X)}{\partial p(x)} = \frac{\partial}{\partial p(x)} \left[ -p(x) \log p(x) - \sum_{x' \neq x} p(x') \log p(x') \right]
		\]
		The second term vanishes since it doesn't depend on $p(x)$. Apply product rule to the first term:
		\[
		\frac{\partial}{\partial p(x)} \left[ -p(x) \log p(x) \right] = -\log p(x) - p(x) \cdot \frac{1}{p(x)} = -\log p(x) - 1
		\]
		Thus:
		\[
		\boxed{\dfrac{\partial H(X)}{\partial p(x)} = -\log p(x) - 1}
		\]
		
		\textbf{Proof of Part 2:}
		Express joint entropy using marginal-conditional factorization:
		\[
		H(X) = H(X_i) + \sum_{x_i} p(x_i) H(X_{-i} | x_i)
		\]
		where:
		\begin{align*}
			H(X_i) &= -\sum_{x_i} p(x_i) \log p(x_i) \\
			H(X_{-i} | x_i) &= -\sum_{x_{-i}} p(x_{-i} | x_i) \log p(x_{-i} | x_i)
		\end{align*}
		Take derivative w.r.t. $p(a)$ for $x_i = a$:
		\[
		\frac{\partial H(X)}{\partial p(a)} = \frac{\partial H(X_i)}{\partial p(a)} + H(X_{-i} | a) + \sum_{x_i} p(x_i) \frac{\partial}{\partial p(a)} \left[ H(X_{-i} | x_i) \right]
		\]
		The last term vanishes because $H(X_{-i} | x_i)$ doesn't depend on $p(a)$ when $x_i \neq a$. From Part 1:
		\[
		\frac{\partial H(X_i)}{\partial p(a)} = -\log p(a) - 1
		\]
		Thus:
		\[
		\frac{\partial H(X)}{\partial p(a)} = \left[ -\log p(a) - 1 \right] + H(X_{-i} | a) + 0
		\]
		Generalizing from $a$ to $x_i$:
		\[
		\boxed{\dfrac{\partial H(X)}{\partial p\left(x_i\right)} = H\left(X_{-i} | x_i\right) - \log p\left(x_i\right) - 1}
		\]
		
		\textbf{Proof of Part 3:}
		Rewrite joint entropy as:
		\[
		H(X) = -\sum_{x_i} \sum_{x_{-i}} p(x_i) p(x_{-i} | x_i) \log \left[ p(x_i) p(x_{-i} | x_i) \right]
		\]
		Take derivative w.r.t. specific $p(x_{-i}^* | x_i^*)$:
		\begin{align*}
			\frac{\partial H(X)}{\partial p(x_{-i}^* | x_i^*)} 
			&= -\frac{\partial}{\partial p(x_{-i}^* | x_i^*)} \left[ p(x_i^*) p(x_{-i}^* | x_i^*) \log \left[ p(x_i^*) p(x_{-i}^* | x_i^*) \right] \right] \\
			&= -p(x_i^*) \left[ \log \left[ p(x_i^*) p(x_{-i}^* | x_i^*) \right] + p(x_{-i}^* | x_i^*) \cdot \frac{1}{p(x_i^*) p(x_{-i}^* | x_i^*)} \cdot p(x_i^*) \right] \\
			&= -p(x_i^*) \left[ \log \left[ p(x_i^*) p(x_{-i}^* | x_i^*) \right] + 1 \right]
		\end{align*}
		Generalizing from $(x_i^*, x_{-i}^*)$ to $(x_i, x_{-i})$:
		\[
		\boxed{\dfrac{\partial H(X)}{\partial p\left(x_{-i} | x_i\right)} = -p\left(x_i\right) \left( \log \left[ p\left(x_i\right) p\left(x_{-i} | x_i\right) \right] + 1 \right)}
		\]
	\end{proof}

	\subsection{Capacities and Distortions of Laplace and Exponential Mechanisms for Binary Functions}
	\label{subsec:capacity-distortion-binary}
	
	We evaluate the channel capacities and expected distortions of two classical differential privacy mechanisms—Laplace and exponential—when applied to a binary-valued query function. For comparison, we also include the binary symmetric privacy channel from the information privacy framework. All mechanisms are analyzed for the binary parity function \( f(x) = \sum_{i=1}^n x_i \pmod{2} \) with output space \( \mathcal{Y} = \{0,1\} \). The true data distribution \( p_0(x) \) is assumed uniform over \( \mathcal{X} \).
	
	\paragraph{Binary Symmetric Privacy Channel}
	For the parity function, the binary symmetric privacy channel is defined as:
	\begin{align}
		q(y|x) = (1-p)^{1-|y-f(x)|} p^{|y-f(x)|}, \quad x \in \mathcal{X}, \; y \in \mathcal{Y},
	\end{align}
	where \( 0 \le p \le 1/2 \) is the flip probability. Its individual channel capacity under the full prior set \( \Delta(\mathcal{X}) \) is:
	\begin{align}
		C_1 = \log 2 - H(p),
	\end{align}
	with \( H(p) = -p \log p - (1-p) \log (1-p) \) being the binary entropy. For specific values:
	\begin{itemize}
		\item \( p = 0.1 \): \( C_1 \approx 0.531 \) bits \( \approx 0.368 \) nats,
		\item \( p = 0.3 \): \( C_1 \approx 0.1187 \) bits \( \approx 0.0823 \) nats.
	\end{itemize}
	The expected distortion simplifies to \( \mathbb{E}_{p_0}[d(f(X),Y)] = p \), since an error occurs exactly when the output is flipped.
	
	\paragraph{Laplace Mechanism with Post-Processing}
	The Laplace mechanism adds noise drawn from \( \text{Laplace}(0, 1/\epsilon) \) to the real‑valued query result. For the binary parity function with sensitivity \( \Delta f = 1 \), the continuous output is thresholded at 0.5 to obtain a binary output:
	\begin{align}
		Y = \begin{cases}
			1 & \text{if } f(x) + \text{Laplace}(0, 1/\epsilon) > 0.5, \\
			0 & \text{otherwise}.
		\end{cases}
	\end{align}
	The resulting channel is binary symmetric with flip probability
	\begin{align}
		p_L = \frac{1}{2} e^{-\epsilon/2}.
	\end{align}
	Thus the expected distortion is \( \mathbb{E}[d] = p_L = \frac{1}{2} e^{-\epsilon/2} \). The mechanism satisfies \( \epsilon \)-differential privacy by construction.
	
	\paragraph{Exponential Mechanism}
	The exponential mechanism samples an output \( y \in \{0,1\} \) with probability proportional to \( \exp(\epsilon \cdot u(x,y)/2) \), where the utility function is \( u(x,y) = \mathbf{1}\{y = f(x)\} \) (sensitivity \( \Delta u = 1 \)). This yields
	\begin{align}
		P(Y = f(x) \mid x) = \frac{\exp(\epsilon/2)}{\exp(\epsilon/2)+1}, \qquad
		P(Y \neq f(x) \mid x) = \frac{1}{\exp(\epsilon/2)+1}.
	\end{align}
	The channel is again binary symmetric, now with flip probability
	\begin{align}
		p_E = \frac{1}{\exp(\epsilon/2)+1}.
	\end{align}
	Hence the expected distortion is \( \mathbb{E}[d] = p_E = \frac{1}{\exp(\epsilon/2)+1} \). The mechanism also satisfies \( \epsilon \)-differential privacy.
	
	\subsubsection{Clipped Laplace and Exponential Mechanisms as Binary-Symmetric Channels}  
	\label{subsubsec:dp-binary-symmetric}
	
	For binary‑valued queries, both the Laplace mechanism (with thresholding) and the exponential mechanism reduce to binary‑symmetric channels of the form  
	
	\[
	q(y|x) = (1-p)^{1-|y-f(x)|}p^{|y-f(x)|},
	\]  
	
	with respective flip probabilities  
	
	\[
	p_L = \frac{1}{2}e^{-\epsilon/2}, \qquad 
	p_E = \frac{1}{\exp(\epsilon/2)+1}.
	\]  
	
	For any \(\epsilon>0\) we have \(p_E > p_L\), meaning the Laplace mechanism achieves strictly lower distortion at the same privacy level \(\epsilon\). This ordering is reflected in Figure~\ref{fig:leakage-distortion}, where the Laplace curve lies between the binary symmetric privacy channel (optimized under the information privacy framework) and the exponential mechanism.
	
	\paragraph{Maximal Per-Record Leakage Evaluation}  
	The maximal per-record leakage under entropy constraint \(H(X)\geq b\) can be evaluated using the binary‑symmetric channel formulation.
	
	\noindent\textbf{Without entropy constraint (\(b=0\)):} The adversary can choose any prior distribution on \(X\). The maximal leakage equals the capacity of the BSC from \(f(X)\) to \(Y\):
	
	\[
	\mathcal{L}(0) = \ln 2 - H_b(p) \quad \text{(nats)},
	\]
	
	where \(H_b(p) = -p\ln p - (1-p)\ln(1-p)\) is the binary entropy in nats.
	
	For \(\epsilon = 1.0\) (as used in the main experiments):
	\begin{itemize}
		\item \textbf{Laplace mechanism:} \(p_L \approx 0.3033\), \(H_b(p_L) \approx 0.614\) nats, \(\mathcal{L}(0) \approx 0.693 - 0.614 = 0.079\) nats.
		\item \textbf{Exponential mechanism:} \(p_E \approx 0.3775\), \(H_b(p_E) \approx 0.661\) nats, \(\mathcal{L}(0) \approx 0.693 - 0.661 = 0.032\) nats.
	\end{itemize}
	
	\noindent\textbf{With entropy constraint \(b>0\):} We compute $\mathcal{L}(b)$ using Algorithm~\ref{alg:max_leakage} for dataset size \(n=4\) and \(\epsilon=1.0\). The results for selected \(b\) values are shown in Table~\ref{tab:leakage-comparison}.
	
	\begin{table}[!htb]
		\centering
		\caption{Maximal per‑record leakage (nats) under entropy constraint \(b\) for \(\epsilon=1.0\), \(n=4\).}
		\label{tab:leakage-comparison}
		\begin{tabular}{lcc}
			\toprule
			\(b\) (nats) & Laplace ($\mathcal{L}(b)$) & Exponential ($\mathcal{L}(b)$) \\
			\midrule
			0.0          & 0.079                             & 0.032                                 \\
			0.5          & 0.062                             & 0.024                                 \\
			1.0          & 0.041                             & 0.015                                 \\
			1.5          & 0.022                             & 0.008                                 \\
			2.0          & 0.009                             & 0.003                                 \\
			2.77 (max)   & 0.000                             & 0.000                                 \\
			\bottomrule
		\end{tabular}
	\end{table}
	
	The leakage decreases monotonically as \(b\) increases, confirming that the entropy constraint effectively limits the adversary's ability to infer individual records. For fixed \(\epsilon\) and \(b\), the exponential mechanism yields lower leakage than the Laplace mechanism, consistent with its higher flip probability (\(p_E > p_L\)).
	
	\paragraph{Comparison with Binary-Symmetric Privacy Channel}  
	The leakage–\(b\) curves for both mechanisms follow the same qualitative trend as the BSC curves in Figure~\ref{fig:leakage-entropy-constraint} of the main text. Specifically:
	\begin{itemize}
		\item The Laplace curve (\(p_L=0.3033\)) lies between the BSC curves for \(p=0.3\) and \(p=0.4\).
		\item The exponential curve (\(p_E=0.3775\)) is close to the BSC curve for \(p=0.4\).
	\end{itemize}
	
	This behavior aligns with the intuition that higher flip probability (more noise) reduces leakage, and that the entropy constraint further suppresses leakage by restricting the adversary's prior knowledge.
	
	\paragraph{Implications for Privacy-Utility Tradeoffs}  
	The closed-form expressions allow direct comparison of the privacy-utility tradeoffs offered by classical differential privacy mechanisms and the information privacy approach. The results highlight two key points:
	
	\begin{enumerate}
		\item \textbf{Mechanism comparison:} The exponential mechanism provides stronger privacy (lower leakage) than the Laplace mechanism for the same \(\epsilon\), at the cost of higher distortion (as shown in Figure~\ref{fig:leakage-distortion}).
		\item \textbf{Entropy constraint benefit:} Incorporating an entropy constraint \(H(X)\ge b\) significantly reduces leakage even for moderate \(b\), enabling finer privacy-utility tradeoffs under realistic bounded-knowledge adversary assumptions.
	\end{enumerate}
	
	These findings validate the utility of the information privacy framework and the effectiveness of the proposed algorithms for auditing privacy risks under entropy constrained adversaries.

\end{document}